\documentclass{article}
\usepackage[utf8]{inputenc}

\usepackage{authblk}
\usepackage[ruled,vlined,linesnumbered]{algorithm2e}
\usepackage{amssymb}
\usepackage{amsmath}
\usepackage{amsthm}
\usepackage{mathtools}
\mathtoolsset{showonlyrefs,showmanualtags}
\usepackage{xspace}
\usepackage{nicefrac}
\usepackage{graphicx} 
\usepackage{caption,subcaption}
\usepackage{natbib}
\usepackage{float}
\usepackage{multirow}

\usepackage[hidelinks]{hyperref}

\usepackage{etoc}

\usepackage{xcolor} 

\DeclarePairedDelimiter\setb\lbrace\rbrace
\DeclarePairedDelimiter\paren\lparen\rparen
\DeclarePairedDelimiter\bracket\lbrack\rbrack

\DeclarePairedDelimiter\norm\lVert\rVert
\DeclarePairedDelimiter\abs\lvert\rvert
\DeclarePairedDelimiterXPP\prob[1]{\Pr}{\lparen}{\rparen}{}{#1}
\DeclarePairedDelimiterXPP\Exp[1]{\mathbb{E}}{\lbrack}{\rbrack}{}{#1}

\DeclareMathOperator{\Varm}{Var}
\DeclarePairedDelimiterXPP\Var[1]{\Varm}{\lparen}{\rparen}{}{#1}

\DeclareMathOperator{\Covm}{Cov}
\DeclarePairedDelimiterXPP\Cov[1]{\Covm}{\lparen}{\rparen}{}{#1}

\DeclareMathOperator{\trm}{tr}
\DeclarePairedDelimiterXPP\tr[1]{\trm}{\lparen}{\rparen}{}{#1}

\DeclarePairedDelimiterXPP\bigO[1]{\mathcal{O}}{\lparen}{\rparen}{}{#1}
\DeclarePairedDelimiterXPP\littleO[1]{o}{\lparen}{\rparen}{}{#1}
\DeclarePairedDelimiterXPP\bigOmega[1]{\Omega}{\lparen}{\rparen}{}{#1}

\newcommand{\tran}{\intercal}

\newcommand{\mat}[1]{\boldsymbol{#1}}
\renewcommand{\vec}[1]{\boldsymbol{#1}}

\newcommand{\Reals}{\mathbb{R}}
\newcommand{\onevec}{\vec{1}}
\newcommand{\zerovec}{\vec{0}}
\newcommand{\idM}{\mat{I}}

\newcommand{\indep}{\perp \!\!\! \perp}

\newtheorem{thm}{Theorem}[section]
\newtheorem{corollary}[thm]{Corollary}

\newtheorem{prop}[thm]{Proposition}
\newtheorem{lemma}[thm]{Lemma}
\newtheorem{assumption}{Assumption}

\newcommand{\divunits}{V_d}
\newcommand{\outunits}{V_o}

\newcommand{\zv}{\vec{z}}
\newcommand{\zve}[1]{z_{#1}}
\newcommand{\zi}{\zve{i}}
\newcommand{\zj}{\zve{j}}
\newcommand{\zl}{\zve{\ell}}

\newcommand{\posym}{Y}

\newcommand{\pove}[1]{\posym_{#1}}
\newcommand{\poi}{\pove{i}}
\newcommand{\poj}{\pove{j}}

\newcommand{\slopee}[1]{\beta_{#1}}
\newcommand{\slopei}{\slopee{i}}
\newcommand{\slopej}{\slopee{j}}
\newcommand{\incepte}[1]{\alpha_{#1}}
\newcommand{\incepti}{\incepte{i}}
\newcommand{\inceptj}{\incepte{j}}

\newcommand{\neigh}{\mathcal{N}}

\newcommand{\dosev}{\vec{x}}
\newcommand{\dosee}[1]{x_{#1}}
\newcommand{\dosei}{\dosee{i}}
\newcommand{\dosej}{\dosee{j}}

\newcommand{\atesym}{\tau}
\newcommand{\ate}{\atesym} 

\newcommand{\itee}[1]{\atesym_{#1}}
\newcommand{\itei}{\itee{i}}
\newcommand{\itej}{\itee{j}}

\newcommand{\dhtest}{\hat{\atesym}}
\newcommand{\dhteste}[1]{\dhtest_{#1}}
\newcommand{\dhtesti}{\dhteste{i}}
\newcommand{\dhtestj}{\dhteste{j}}

\newcommand{\wij}{w_{i,j}}
\newcommand{\weightM}{\mat{W}}

\newcommand{\divdeg}{d_d} 
\newcommand{\outdeg}{d_o} 
\newcommand{\depN}[1]{\mathcal{I}(#1)} 
\newcommand{\depNi}{\depN{i}}
 
\newcommand{\depdeg}{D} 
\newcommand{\maxcorr}{k} 
\newcommand{\erre}[1]{a_{#1}} 
\newcommand{\erri}{\erre{i}}
\newcommand{\errj}{\erre{j}}
\newcommand{\maxpo}{M} 

\newcommand{\clustsym}{C}
\newcommand{\cluste}[1]{\clustsym_{#1}}

\newcommand{\clustering}{\mathcal{\clustsym}}

\newcommand{\doseht}{\textsc{ERL}\xspace}
\newcommand{\exposuredesign}{\textsc{Exposure-Design}\xspace}
\newcommand{\corrclust}{\textsc{Corr-Clust}\xspace}
\newcommand{\corrclustCS}{\textsc{Corr-Clust-CS}\xspace}

\newcommand{\erlvarest}{\widehat{\mathrm{Var}}(\dhtest)}
\newcommand{\covest}[1]{\widehat{C}_{#1}}
\newcommand{\covestij}{\covest{i,j}}
\newcommand{\covestkl}{\covest{k,\ell}}
\newcommand{\covmat}{\Sigma}
\newcommand{\covmatij}{\covmat_{i,j}}
\newcommand{\covmatii}{\covmat_{i,i}}

\newcommand{\cerrij}{\erre{i,j}}
\newcommand{\cerrkl}{\erre{k,\ell}}

\newcommand{\depNvar}[1]{\mathcal{I}_V(#1)} 
\newcommand{\depdegvar}{D_V}

\newtheorem{manualtheoreminner}{Theorem}
\newenvironment{manualtheorem}[1]{%
\renewcommand\themanualtheoreminner{#1}%
\manualtheoreminner
}{\endmanualtheoreminner}

\newtheorem{manualpropinner}{Proposition}
\newenvironment{manualprop}[1]{%
\renewcommand\themanualpropinner{#1}%
\manualpropinner
}{\endmanualpropinner}

\newtheorem{manualcorrinner}{Corollary}
\newenvironment{manualcorr}[1]{%
\renewcommand\themanualcorrinner{#1}%
\manualcorrinner
}{\endmanualcorrinner}

\title{Design and Analysis of Bipartite Experiments \\ under a Linear Exposure-Response Model}
\author[1]{Christopher Harshaw}
\author[1]{Fredrik S{\"a}vje}
\author[2]{\\David Eisenstat}
\author[2]{Vahab Mirrokni}
\author[2]{Jean Pouget-Abadie}

\affil[1]{Yale University}
\affil[2]{Google Research}
\date{\today}

\begin{document}

\pagenumbering{gobble}

\makeatletter%
\begin{NoHyper}\gdef\@thefnmark{}\@footnotetext{\hspace{-1em}We thank 
P.M. Aronow,
Kay Brodersen,
Nick Doudchenko,
Ramesh Johari,
Khashayar Khosravi,
Sebastien Lahaie,
Vahan Nanumyan,
Georgia Papadogeorgou,
Lewis Rendell,
Johan Ugander,
and
C. M. Zigler
for stimulating discussions which helped shape this work.
Christopher Harshaw gratefully acknowledges support from an NSF Graduate Research Fellowship (DGE1122492) as well as support from Google as a Summer Intern and a Student Researcher.}\end{NoHyper}%
\makeatother%

\maketitle

\begin{abstract}
    A bipartite experiment consists of one set of units being assigned treatments and another set of units for which we measure outcomes.
The two sets of units are connected by a bipartite graph, governing how the treated units can affect the outcome units.
In this paper, we consider estimation of the average total treatment effect in the bipartite experimental framework under a linear exposure-response model.
We introduce the Exposure Reweighted Linear (\doseht) estimator, and show that the estimator is unbiased, consistent and asymptotically normal, provided that the bipartite graph is sufficiently sparse.
To facilitate inference, we introduce an unbiased and consistent estimator of the variance of the \doseht point estimator.
In addition, we introduce a cluster-based design, \exposuredesign, that uses heuristics to increase the precision of the \doseht estimator by realizing a desirable exposure distribution.
\end{abstract}

\newpage

\etocsettagdepth{mtappendix}{none}
\tableofcontents

\newpage

\pagenumbering{arabic}
\section{Introduction}\label{sec:intro}

Two-sided marketplaces are rife with interesting but difficult causal questions. 
What happens to demand if shipping times or fees are reduced? 
What happens to people’s willingness to use ride-hailing apps if more drivers are enrolled in specific cities? 
What happens to long term user behavior if a hotel booking platform changes its recommendation engine?
The causal impact of these changes is hard to quantify, even when using randomized experiments, because marketplace dynamics often violate a central tenet of conventional experimentation: the Stable Unit Treatment Value Assumption, abbreviated SUTVA.
This assumption stipulates that the treatment assigned to one unit does not affect any other units. Violations of this assumption is a phenomenon known as interference, which is often present in the case of marketplace experiments and complicates causal analysis.

The bipartite experimental framework offers a useful formalism to study two-sided market experiments and other violations of SUTVA that can happen along the edges of a bipartite graph. This stands in contrast with interference that occurs on graphs where all units are of the same type (e.g. users of a social network). 
In the bipartite experimental framework, we distinguish two types of units: units that can be subject to an intervention and units whose responses are of interest to the experimenter. 
We assign treatment to the former and measure the outcomes of the latter. The causal impact of treating one group of units is measured on the other group by tracking the \emph{exposure} to treatment that the latter group receives, informed by the knowledge of the bipartite graph between them.
We remark that the treatment status of a single unit may affect the measured outcomes of many units and, likewise, a measured outcome may be affected by many treatment units.

As an example, consider a marketplace where buyers compete for limited goods, some of which may be perfectly or partially substitutable. 
Their demand of these goods form a bipartite graph that potentially can be inferred by the marketplace owner. 
The owner of the marketplace would like to determine the causal effect of discounting prices on buyers' marketplace behavior through a randomized experiment. 
Randomly assigning certain buyers to receive a discounted price is often not possible, and might even be prohibited, in which case randomization is only possible at the item-level. 
At the same time, simply comparing discounted goods with non-discounted goods runs the risk of severe bias: a discounted good may do well against a non-discounted substitutable good, which does not accurately reflect a world where either both or neither are discounted. 
To address this, the marketplace owner decides to monitor this change at the buyer level.
By tracking both buyers' behavioral changes and exposure to discounted goods, the causal effect of the discount can then be teased out.

As is done in much of the interference literature and other settings where SUTVA is violated, assumptions on potential outcomes are made when the bipartite graph has a many-to-many structure in order to allow for tractable inference. 
One such assumption is existence of an exposure mapping, which posits that outcomes are some simple function of the treatment assignments of neighboring units in the bipartite graph \citep{Toulis13, Aronow2017}. 
In this work, we study estimation of an all-or-nothing treatment effect in the bipartite experimental framework under a linear exposure-response model, where exposures are linear functions of assignments and responses are linear functions of the exposures.
The main contributions of this work are summarized as follows:
\begin{itemize}
    \item We describe the Exposure-Reweighted Linear (\doseht) estimator, an unbiased linear estimator of the average total treatment effect under the linear exposure-response model. 
    We show that the \doseht estimator is consistent and asymptotically normal, provided the graph is sufficiently sparse.
    \item We describe a variance estimator, which may be used to construct confidence intervals via a normal approximation. 
    We show that under mild conditions on the exposure distribution, the variance estimator is unbiased and consistent.
    Additionally, we prove asymptotic validity of normal-based confidence intervals using the variance estimator.
    \item We describe the behavior of the \doseht estimator when the linear exposure-response model does not hold, showing that it still estimates an interpretable and policy relevant causal quantity.
    \item We describe \exposuredesign, a cluster-based design which aims to increase the precision of the \doseht estimator.
    The design achieves this by increasing the variance of individual exposures while decreasing the covariance between different exposures.
    This improves precision is several settings of interest.
\end{itemize}

\paragraph{Organization} The remainder of the paper is organized as follows:
We review related work in Section~\ref{sec:related-works}.
In Section~\ref{sec:experimental-setting}, we present the bipartite experimental setting and the linear exposure-response model.
We present and analyze the \doseht estimator in Section~\ref{sec:estimator} and its corresponding variance estimator and confidence intervals in Section~\ref{sec:confidence-intervals}.
In Section~\ref{sec:no-assumption}, we analyze the \doseht estimator when the linear response assumption does not hold.
We present our cluster-based design \exposuredesign in Section~\ref{sec:design}.
Finally, we apply this methodology on the publicly available Amazon product review graph are presented in Section~\ref{sec:simulations}, and we conclude in Section~\ref{sec:conclusion}.

\section{Related Works}\label{sec:related-works}

Within the wide-ranging causal inference literature, our work falls squarely within the subset relying on the potential outcomes framework~\citep{neyman1923applications, imbens2015causal}. The design and analysis of randomized experiments in the presence of interference has garnered plenty of attention, spanning vaccination trials~\citep{struchiner1990behaviour}, agricultural studies~\citep{kempton1997interference}, voter-mobilization field experiments~\citep{sinclair2012detecting}, and viral marketing campaigns \citep{aral2011creating, eckles2016estimating}.  It is beyond the scope of the current paper to extensively review the literature causal inference with interference. Instead, we direct readers to the review article by \citet{halloran2016dependent}.

Our work is primarily motivated by marketplace experiments. Evidence of interference in marketplaces has been noted across industries for various experimental designs~\citep{gupta2019top}. \citet{reiley2006field}, \citet{einav2011learning} and \citet{holtz2020reducing} study the interference bias that results from supply-side randomization, while \citet{blake2014marketplace} and \citet{fradkin2015search} consider this problem in the context of demand/user-side randomization. \citet{basse2016randomization} and \citet{liu2020trustworthy} compare supply-side randomization to two-sided randomization as well as to budget-split designs, showing bias can be reduced in the context of certain ad auction experiments. More recently, \citet{johari2020experimental} characterize which randomization scheme (supply-side, demand-side, or two-sided) leads to reduced bias as a function of market balance.

We consider a slightly different experimental setting, introduced by~\citet{zigler2021bipartite}, characterized by random assignment of treatment on one side of the bipartite graph (demand- or supply-side), while outcomes are measured on the other side. 
The advantage of this framework is that the bipartite graph defines an exposure function (similar to \citet{Aronow2017}), which is assumed to solely determine an unit's outcome.
\citet{zigler2021bipartite} study causal estimands which are more closely related to direct effects rather than the all-or-nothing treatment effect considered here.
A main drawback of their work is that the analysis of their estimators requires that the bipartite graph be the union of many small connected components. 

To allow for more complex bipartite graphs, methodologists opt for stronger structural assumptions on the outcomes.
An exposure-response assumption similar to the one we use is adopted by
\citet{pouget2019variance}, who introduce a cluster-based design for general bipartite graphs, consider a similar estimand, and are also motivated by marketplace experiments.
Later, \citet{doudchenko2020causal} proposed a class of generalized propensity score estimators for this framework, which are unbiased for both experimental and observational settings under standard assumptions and a similar exposure-response assumption. 

Our work is the first to propose methods for provably valid inference (e.g., confidence intervals) in the bipartite settings and to jointly consider estimators and designs which improve overall precision of treatment effect estimators.
While the cluster design of \citet{pouget2019variance} is based on the intuition of achieving a large spread of exposures, it disregards the correlation of exposures and is not directly tied to the performance of an estimator.
Additionally, while the estimators proposed by \citet{doudchenko2020causal} are unbiased, they are based on a different approach which requires fitting a generalized propensity score function. 
Neither of these papers present methods for valid inference.

\section{Experimental Setting}\label{sec:experimental-setting}
In the bipartite experimental framework, the units which receive treatment are distinct from the units on which the outcomes are measured.
For example, \cite{zigler2021bipartite} apply the framework to analyze how interventions on power plants' pollution affect the hospitalization rates among nearby hospitals.
We discuss the general bipartite framework in Section~\ref{sec:bipartite-experiments} and the linear exposure response assumption in Section~\ref{sec:linear-exposure-response}.

\subsection{Bipartite experiments}\label{sec:bipartite-experiments}
In the bipartite experiment setting, there are two groups of units: the \emph{diversion units}, to which treatment is applied, and the \emph{outcome units}, where outcomes are measured.
We denote the set of $m$ diversion units by $\divunits$ and the set of $n$ outcome units by $\outunits$.

Each of the $m$ diversion units receives a (random) binary treatment $\zi \in \setb{0,1}$, and we collect these treatments into a treatment vector, $\zv = \paren{\zve{1}, \zve{2}, \dots \zve{m}} \in \setb{0,1}^m$.
The distribution over the random treatment vectors is called the \emph{design} of the experiment and it is chosen by the experimenter.
Each of the outcome units $i \in \outunits$ is associated with a potential outcome function $\poi(\zv)$, which maps the treatment assignments to the observed value, which is a real number.
In the bipartite setting, we assume that each potential outcome function depends only on the treatment of a neighborhood set of diversion units.
More formally, there exists a \emph{neighborhood mapping} $\neigh: \outunits \rightarrow 2^{\divunits}$ such that for all outcome units $i \in \outunits$,
\[
\poi(\zv) = \poi(\zv') 
\quad \text{ if }
\zj = \zj' \text{ for all } j \in \neigh(i) \enspace.
\]
Throughout the paper, we assume that the neighborhood mapping is known and correctly specified, so that the above condition holds. 
We recover the standard Stable Unit Treatment Value Assumption (SUTVA) when the diversion units are identified with the outcome units and the neighborhood mapping is the identity function.

The number of potential outcomes for each outcome unit grows exponentially in the size of its neighborhood.
\cite{zigler2021bipartite} avoid this issue by assuming that the bipartite structure is the union of many small connected components.
Unfortunately, this is typically not a reasonable assumption in the marketplace settings where we know that more varied interactions occur: buyers may interact with a variety of products.
Without further restrictions on the structure of the neighborhoods or the potential outcome functions, inference of any causal estimand is impossible \citep{Basse2018Limitations, Savje2021}.
Take, for example, an instance where the neighborhood of each outcome unit is all diversion units.
For this reason, we restrict our attention to settings where a stronger assumption on the potential outcomes is reasonable.

\subsection{Linear exposure-response model}\label{sec:linear-exposure-response}
In order to admit tractable inference of causal estimands, we consider a linear exposure-response model, which consists of two underlying assumptions: a linear exposure assumption and a
linear response assumption, which we state formally below.

In the linear exposure-response model, we suppose that there is a weighted bipartite graph between diversion units and outcomes units, where the edges have non-negative weights $\wij \geq 0$, which we arrange into an $n$-by-$m$ incidence matrix $\weightM$.
An edge $\wij$ represents the influence of diversion unit $j$ on the outcome units $i$.
We say that outcome unit $i$ and diversion unit $j$ are \emph{incident} if the weight $w_{i,j}$ is positive.
The degree of a diversion unit is the number of outcome units it is incident to, and the largest degree among all diversion units is denoted $\divdeg$.
The degree of an outcome unit is defined similarly and the largest degree among all outcome units is denoted $\outdeg$.
For simplicity, we assume that each outcome unit has degree at least 1 and the weights incident to an outcome unit are normalized to sum to one, i.e. the rows of the incidence matrix $\weightM$ sum to one.
We also assume that this weighted bipartite graph is known to the experimenter.
In many market experiments, the experimenter may construct an approximation of this graph from historical data.

The \emph{linear exposure assumption} is that the treatment assignments influences the potential outcomes only through a linear combination, which is more structured than arbitrary influence.
More formally, for each outcome unit $i \in \outunits$, the \emph{exposure} of outcome unit $i$ is
\[
\dosei(\zv) = \sum_{j \in \divunits} \wij \zj \enspace,
\]
and for all pairs of assignment vectors $\zv$ and $\zv'$ with $\dosei(\zv) = \dosei(\zv')$, we have that $\poi(\zv) = \poi(\zv')$.
This implies that the neighborhood mapping is such that $\neigh(i) = \{j : \wij > 0\}$.

We arrange these $n$ exposures into an exposure vector $\dosev(\zv) = \paren{\dosee{1}(\zv), \dosee{2}(\zv) \dots \dosee{n}(\zv) }$.
Because the exposure is a function of treatment, the experimental design completely determines the exposure distribution.
This linear exposure assumption is a generalization of the partial and stratified interference assumptions discussed by \citet{hudgens2008toward}.
When the treatment assignment vector $\zv$ is clear from context, we write simply $\dosei$ and $\dosev$ for the $i$th exposure and the exposure vector, respectively.
Similarly, we write $\poi$ for the outcomes.
Using matrix-vector notation, we may write the exposure vector as $\dosev(\zv) = \weightM \zv$.
Due to the normalization of the weights and the binary values of the treatment assignment, each exposure takes values in the range $[0,1]$.

The \emph{linear response assumption} is that for each outcome unit, the potential outcome is a linear function of its exposure.
That is, for each outcome unit $i \in \outunits$, there exists parameters $\incepti$ and $\slopei$ such that
\[
\poi(\zv) = \incepti + \slopei \dosei(\zv) \enspace.
\]
We refer to $\incepti$ as the unit-specific intercept and $\slopei$ as the unit-specific slope.
Note that the linear function does not need to be the same between units.
These terms are unknown to the experimenter, and the experimenter only observes the outcome $\poi(\zv)$, along with the sampled assignment vector $\zv$ and the resulting exposure vector $\dosev$.

We refer to the \emph{linear exposure-response model} as the combination of the linear exposure assumption and
the linear response assumption.
The linear exposure-response model places certain limits on the potential outcomes, but allows for more complex structure in the bipartite graph than previous work.
This trade-off is preferable in settings such as marketplace experiments, where we know that a complex bipartite structure exists and we are more comfortable with making simplifying assumptions about potential outcomes.
For further discussion on empirical and theoretical evidence for complex structure in marketplace experiments, we refer the reader to \citet{blake2014marketplace, Andrey2017Search, johari2020experimental}.

Structural assumptions on the outcomes similar to the linear exposure-response assumption presented here are commonly made throughout the interference literature.
The \emph{linear-in-means} (LIM) model posits that a unit's response is a linear function of their own treatment, and the mean of the treatments of their group \citep{Manski1993}.
The LIM model has been extended in various ways in the context of partial interference \citep{Baird2018, OfferWestort2021} and social network experiments \citep{Bramoulle009, Toulis13}.
\citet{Chin2019} investigates the use of machine learning estimators for the global average treatment effect under a variation of the LIM when the terms in the linear model of arbitrary functions of treatment.
\citet{basse2016randomization} study model-assisted estimators and designs under the ``normal sum-model'' which is similar to the linear exposure-response considered here, but with a normal noise term.
We remark that the bipartite setting with the linear exposure-response assumption recovers the standard SUTVA setting when diversion units are identified with the outcome units and the weight matrix is the identity.

From one perspective, the linear exposure-response model is a strong assumption. 
It requires that the response for each unit is exactly a linear function in the exposure. 
This rules out, for example, that different diversion units have different impacts on a single outcome unit. 
But from another perspective, the model is completely unrestrictive: it does not limit the heterogeneity between units at all. 
That is, knowing the response function for one unit tells us nothing about the response function of other units.
While there are few settings in which the linear exposure-response model will hold exactly, it will often be a useful approximation given its unrestrictiveness with respect to heterogeneity. 
In Section~\ref{sec:no-assumption}, we analyze the behavior of the \doseht estimator under a general non-linear response assumption, finding that it estimates a best linear approximation to the average response. 
However, we leave it to future work to more finely characterize the behavior of estimator under general responses and we assume the linear exposure-response model holds exactly throughout the paper, unless otherwise stated.

\subsection{Causal estimand}\label{sec:causal-estimand}
We are interested in understanding the contrast between two possible worlds: one where all diversion units receive treatment and one where they all receive control.
For an individual outcome unit, this contrast is captured by the individual treatment effect, $\itei = \poi(\zv = \onevec) - \poi(\zv = \zerovec)$ for $i \in \outunits$.
Just as in the typical SUTVA setting, we cannot hope to estimate the individual treatment effects well because only one potential outcome is observed for any one unit.
In light of this, we opt to estimate an aggregated causal quantity.
In this paper, we are interested in the Average Total Treatment Effect (ATTE), which is the average contrast between the scenario that all diversion units receive treatment and all diversion units receive control. More precisely, ATTE is defined as
\[
\ate 
= \frac{1}{n} \sum_{i=1}^n \itei
= \frac{1}{n} \sum_{i=1}^n \bracket[\Big]{\poi(\zv = \onevec) - \poi(\zv = \zerovec)}
\]

Under the linear exposure-response assumption, the ATTE is proportional to the average of the slope terms, as shown in the following proposition. 

\begin{prop}\label{prop:ate-under-linear-exposure-response}
Under the linear exposure-response assumption, the ATTE is 
$ \ate = \frac{1}{n} \sum_{i=1}^n \slopei $.
\end{prop}
\begin{proof}
The individual treatment effect of outcome unit $i$ is equal to its slope, as
\[
\itei 
= \poi(\zv = \onevec) - \poi(\zv = \zerovec) \\
= \bracket[\big]{ \slopei \dosei(\zv = \onevec) + \incepti} - \bracket[\big]{ \slopei \dosei(\zv = \zerovec) + \incepti}
= \slopei \enspace,
\]
where we have used that $\dosei(\zv = \onevec) = 1$ and $\dosei(\zv = \zerovec) = 0$.
The result follows by taking the average of the individual treatment effects.
\end{proof}

There are two main challenges in estimating the ATTE in this setting: we want to estimate the average of the slopes of many different linear response functions, but only one point from each of the distinct linear response functions is observed.
Although stated in somewhat unfamiliar terms, this is the fundamental problem of causal inference \citep{Holland84statistics}.
The second challenge is that of constructing a treatment design which realizes a desirable exposure distribution.
This is a difficult task when the bipartite weight matrix has non-trivial overlapping structures.
In the remainder of the paper, we focus on addressing these two challenges by developing an estimator and a class of designs which together accurately estimate the ATTE. 

\subsection{Cluster designs}\label{sec:cluster-designs}

Some of the analysis in this paper assumes that the treatment is assigned according to a \emph{independent cluster design}, where the diversion units are grouped into clusters and treatment is assigned to an entire cluster.
More formally, we say that a partition $\cluste{1}, \cluste{2}, \dots, \cluste{\ell}$ of the diversion units is a \emph{clustering}, which we denote as $\clustering = \setb{\cluste{1}, \cluste{2}, \dots, \cluste{\ell}}$.
That is, all clusters are disjoint and the union of all clusters is set of diversion units $\divunits$.
Given a clustering $\clustering$, a treatment assignment from the corresponding \emph{independent cluster design} is drawn in the following way: independently for each cluster, we assign all diversion within a cluster to have either treatment $\zi = 1$ with probability $p$ and treatment $\zi = 0$ with probability $1-p$. 
For notational simplicity, we consider the treatment probability $p$ to be fixed for all clusters, but our results extend to the setting where each cluster has its own treatment probability.
Note that the class of independent cluster designs is completely specified by $\clustering$ and $p$.

\section{The Exposure Reweighted Linear Estimator}\label{sec:estimator}
We introduce the Exposure Reweighted Linear (\doseht) estimator, which is an unbiased estimator of the ATTE under the linear exposure-response assumption, defined as
\begin{equation} \label{eq:ht-exposure-def}
\dhtest 
= \frac{1}{n} \sum_{i=1}^n \poi \paren[\bigg]{ \frac{ \dosei - \Exp{\dosei}}{\Var{\dosei}} }
\enspace.
\end{equation}
The \doseht estimator requires knowledge of the mean and variance of each of the marginal exposure distributions under the treatment design.
For several commonly used designs such as Bernoulli and independent cluster designs, these exposure characteristics may be computed directly; however, for arbitrary designs, the expectation and variance of the exposures may need to be estimated to high precision using samples drawn from the treatment design.
We assume here that these exposure characteristics are known exactly.
We emphasize that the \doseht estimator may be used under any treatment design and not just the cluster-based treatment design we propose in Section~\ref{sec:design}.

The \doseht estimator belongs to the class of linear estimators, as it is a (random) linear function of the observed outcomes.
It shares similarities with the style of Horvitz--Thompson estimators \citep{Narain51sampling, HT52generalization}:
Horvitz--Thompson estimators re-weight an outcome by the probability of observing that outcome, while the \doseht estimator re-weights an outcome by the normalized distance of the exposure from its mean.
When there are many possible values of exposures, such as under the linear exposure-response model, the type of re-weighting done by the Horvitz--Thompson estimator would lead to excessively large variance.

\subsection{Unbiasedness and consistency of the \doseht estimator} \label{sec:basic-stats}
In this section, we analyze the behavior of the \doseht estimator as a point estimator of the average total treatment effect (ATTE).
First, we show that the \doseht estimator is unbiased. 
Then we show consistency and asymptotic normality of the \doseht estimator, provided that the bipartite graph is not too dense. 
Theorem~\ref{thm:exposure-ht-unbiased} below ensures that under mild conditions on the treatment design, there is no systematic bias in the \doseht estimator.

\newcommand{\exposurehtunbiased}{
Suppose the design is such that each exposure has positive variance, $\Var{\dosei} > 0$.
Under the linear response assumption, the \doseht estimator is unbiased for the ATTE: $\Exp{\dhtest} = \ate$.
}
\begin{thm}[Unbiasedness] \label{thm:exposure-ht-unbiased}
\exposurehtunbiased
\end{thm}

Next, we analyze the asymptotic behavior of the \doseht estimator.
In the asymptotic analysis, we suppose that there is a sequence of bipartite experiments, in which the number of units is growing to infinity.
Strictly speaking, all quantities of the experiment such as the bipartite graph, the outcomes, the treatment design, etc, should be indexed by an integer $N$; however, we drop these subscripts for notational clarity.

We make two additional assumptions about the bipartite experiments in this asymptotic sequence.
The first is that the potential outcomes are bounded.
The second is that the design has limited dependence between treatment assignments.

\begin{assumption}[Bounded Potential Outcomes]\label{assumption:bounded-potential-outcomes}
The potential outcomes are bounded in absolute value $\abs{ \poi(\zv) } \leq \maxpo$, where $\maxpo$ is a constant.
\end{assumption}

\begin{assumption}[Design Conditions] \label{assumption:design}
The treatments assignments are distributed according to an independent cluster design, where the probability of treatment $p$ is bounded away from $0$ and $1$ by a constant in the asymptotic sequence.
Additionally, the sizes of clusters are bounded by $\maxcorr$, a constant in the asymptotic sequence. 
\end{assumption}

Assumption~\ref{assumption:design} rules out certain classes of treatment designs, such as complete randomization (i.e. group balanced designs).
While it may be possible to obtain similar asymptotic results under such designs, we limit our consideration to those satisfying Assumption~\ref{assumption:design}.
Under these assumptions, we prove that \doseht is consistent when the bipartite graph is not too dense.

\newcommand{\consistencythm}{
Under Assumptions~\ref{assumption:bounded-potential-outcomes} and \ref{assumption:design}, the mean squared error of the \doseht estimator is bounded as
$ \Exp{ \paren{\dhtest - \ate}^2 } = \bigO[\big]{\divdeg \outdeg^3 / n} $.
Thus, the estimator is consistent if $\divdeg \outdeg^3 = \littleO{n}$.
}
\begin{thm}[Consistency] \label{thm:consistency}
\consistencythm
\end{thm}

Theorem~\ref{thm:consistency} shows that the convergence rate of the \doseht estimator is at least $\sqrt{\divdeg \outdeg^3 / n}$, where $\divdeg$ and $\outdeg$ are the maximum degrees of the diversion and outcome units, respectively.
The main technical assumption that we require for consistency is that $\divdeg \outdeg^3 = \littleO{n}$ in the asymptotic sequence.
Informally, this condition ensures that the bipartite graph is not too dense as it grows.
Indeed, we expect the convergence of any estimator to worsen as the graph becomes more complex and dense.
While consistency may hold under weaker conditions for particular designs , an assumption on the graph density must be made:
in the complete bipartite graph where all outcome units receive the same exposure, consistent estimation is impossible.

We now discuss a setting where this condition $\divdeg \outdeg^3 = \littleO{n}$ holds.
Suppose that each diversion unit has fixed degree $\divdeg$, which is a constant with respect to $m$ and $n$.
The average degree of an outcome unit is then $\bar{\outdeg} = \divdeg m/n$.
Assuming that the maximum outcome degree $\outdeg$ is within a constant factor of the average, this yields that the term $ \outdeg = \bigO{m/n}$.
Using that the diversion degrees are constant, we get that $\divdeg \outdeg^3 = \littleO{n}$ if $m = \littleO{n^{4/3}}$.
Thus, in graphs with constant diversion degrees where the edges are roughly evenly distributed between outcome units, the premise of Theorem~\ref{thm:consistency} holds when $m$ grows at a rate slower than $n^{4/3}$.

\subsection{Asymptotic normality of the \doseht estimator}

We next characterize the asymptotic distribution of the estimator.
In particular, we show that the \doseht estimator converges to a normal distribution as the size of the bipartite experiment grows, provided that the graph remains sparse.
This result is derived under the same asymptotic regime as above.
In order to prove the central limit theorem, we require an additional assumption on the asymptotic sequence of bipartite experiments.
Namely, we require that the variance of the \doseht estimator decreases no faster than the parametric rate.

\begin{assumption}\label{assumption:mse-rate}
The normalized variance of the \doseht estimator $n \cdot \Var{\dhtest}$ is bounded away from zero asymptotically.
\end{assumption}

Assumption~\ref{assumption:mse-rate} rules out settings in which we can estimate the ATTE at an unusually fast rate. 
It is theoretically possible to estimate ATTE at a faster than parametric rate, but these settings are not practically relevant. 
In particular, Assumption~\ref{assumption:mse-rate} rules out three scenarios. 
The first is when the magnitude of the potential outcomes approaches zero in the asymptotic sequence. 
This requires that almost all potential outcomes approach zero; the magnitude of the potential outcomes are generally non-zero even when their average is zero. 
The second scenario is when the design almost perfectly pinpoints the potential outcomes. 
This can be formalized as the variance of each individual term of the estimator diminishes asymptotically, i.e. $\Var{\dhtesti} \rightarrow 0$, where $\dhtesti = \poi(\zv) (\dosei - \Exp{\dosei}) / \Var{\dosei}$. 
The third scenario is when the \doseht estimator converges at a parametric rate, but the asymptotic variance happens to be exactly zero.
All of these scenarios are knife-edge cases that we have good reason to believe would not materialize in practice. 
Even if they do, the estimator would still be unbiased and consistent, but its asymptotic distribution might not be normal.

We are now ready to present a central limit theorem which states that under mild regularity conditions, the \doseht estimator is asymptotically normal.

\newcommand{\asymptoticanormalitythm}{
Under Assumptions~\ref{assumption:bounded-potential-outcomes}, \ref{assumption:design}, and \ref{assumption:mse-rate}, and supposing that $\divdeg^4 \outdeg^{10} = \littleO{n}$, the \doseht estimator is asymptotically normal:
\[
\frac{\dhtest - \ate}{\sqrt{\Var{\dhtest}}} \xrightarrow[]{d} \mathcal{N}(0, 1) \enspace.
\]
}
\begin{thm}[Asymptotic Normality] \label{thm:asymptotic-normality}
\asymptoticanormalitythm
\end{thm}

The proof relies on Stein's method for bounding distances between distributions (see, e.g. \citet{ross2011fundamentals}).
We use Stein's method because standard techniques for establishing central limit theorems which heavily rely on independence are not applicable in the bipartite experimental framework where exposures are necessarily correlated.
We remark that Stein's method has been recently used for obtaining limiting behavior of other estimators in the interference literature \citep{Aronow2017, chin2019central, ogburn2020causal}.

The assumptions on the asymptotic growth of the bipartite graph may be interpreted similarly as those appearing in Theorem~\ref{thm:consistency}.
Namely, they prevent the bipartite graph from becoming too dense.
We remark that the growth assumptions required for asymptotic normality (Theorem~\ref{thm:asymptotic-normality}) are stronger than those required required for consistency (Theorem~\ref{thm:consistency}).
The growth conditions in Theorem~\ref{thm:asymptotic-normality} are only sufficient and we conjecture that they are not necessary for asymptotic normality.
One aspect the growth conditions ensure is that the variance of the exposures does not converge to zero at a too fast rate.
If we ensure positivity of the exposures through some other means, then the growth conditions can be weakened.
For example, assuming that $\Var{\dosei} \geq c > 0$, the growth conditions are relaxed to $\divdeg^4 \outdeg^{4} = \littleO{n}$.
Furthermore, the proof of Theorem~\ref{thm:asymptotic-normality} assumes that the \doseht estimator attains a parametric rate of convergence.
This is not always be the case, and accounting for the true convergence rate will generally weaken the growth conditions for asymptotic normality.
However, weakening the growth conditions beyond this would require a different analysis, either by a more careful application of Stein's method or by different means all together.

Assumption~\ref{assumption:design} allows for a broad class of designs.
For example, unit-level Bernoulli randomization falls into this class, but this design does not consider the structure of the bipartite graph and will generally perform poorly.
To derive analytical results for this broad class of designs, the growth conditions on the bipartite graph are quite restrictive, and they may be too restrictive in certain settings where more dense interaction patterns occur.
If one restricts focus to a smaller class of designs, these growth conditions could potentially be weakened.
The key implication of Assumption~\ref{assumption:design} together with the growth conditions is that the variance of the exposures is large and the correlation between most pairs of exposures is small.
Heuristically, these conditions on the exposure distribution are the main aspects required for consistency and normality.
We describe a design in Section~\ref{sec:design} that directly targets the exposure distribution to satisfy these conditions, and it will therefore be better behaved than many of the designs allowed by Assumption~\ref{assumption:design}.

\section{Variance Estimation and Confidence Intervals}\label{sec:confidence-intervals}

In this section, we present methods for constructing confidence intervals for the ATTE in the bipartite setting under the linear exposure-response assumption.
If we knew the variance of the \doseht estimator, we could use Theorem~\ref{thm:asymptotic-normality} directly to construct asymptotically valid confidence intervals.
Because the variance of the \doseht estimator depends on the unobserved potential outcomes, we must construct an estimator of the variance.

In the finite population experimental settings with binary treatments, unbiased variance estimation is not possible without strong assumptions on the heterogeneity between units \citep{imbens2015causal}.
In light of this negative result, experimenters tend to favor conservative variance estimators that lead to valid but overly wide confidence intervals.
In contrast to the typical experimental settings with binary treatments, we show that unbiased variance estimation is possible in the bipartite setting under the linear response assumption when the exposures take many (i.e. more than two) values.

\subsection{Derivation of the variance estimator}

To the best of our knowledge, this approach to constructing a variance estimator is new and so we highlight its derivation here.
Our approach to constructing a variance estimator begins by decomposing the \doseht estimator into a weighted average of individual effect estimators, and to decompose the variance of the \doseht estimator as the average of covariances of these individual effect estimators.
To this end, define $\dhtesti \triangleq \poi \paren{\dosei - \Exp{\dosei}} / \Var{\dosei}$ to be the individual terms in the \doseht estimator.
We may interpret $\dhtesti$ as an unbiased, but very imprecise, estimator of the individual treatment effect $\itei$.
The \doseht estimator can be written as the average of these quantities: $\dhtest = (1/n) \sum_{i=1}^n \dhtesti$.
The variance of the \doseht estimator may be written as
\[
\Var{ \dhtest }
= \Var[\bigg]{ \frac{1}{n} \sum_{i=1}^n \dhtesti }
= \frac{1}{n^2} \sum_{i=1}^n \sum_{j=1}^n \Cov{ \dhtesti, \dhtestj}
 \enspace.
\]

Our approach to constructing a variance estimator will be to construct simple weighting estimators for each of the $\Cov{\dhtesti, \dhtestj}$ terms.
These weighting estimators will be expressed as $\covestij = \poi \poj R_{i,j}(\dosei, \dosej)$, where $\poj \poj$ is the product of observed outcomes and $R_{i,j}(\dosei, \dosej)$ is a weighting function which takes the observed exposures as inputs, such that
our variance estimator will be the simple average of these weighted products of outcomes:
\[
\erlvarest 
\triangleq \frac{1}{n^2} \sum_{i=1}^n \sum_{j=1}^n \covestij
\triangleq \frac{1}{n^2} \sum_{i=1}^n \sum_{j=1}^n \poi \poj R_{i,j}(\dosei, \dosej) \enspace.
\]
The goal is to make each individual estimator $\covestij$ unbiased for each individual covariance $\Cov{\dhtesti, \dhtestj}$, so that the entire variance estimator $\erlvarest$ will be unbiased.
Moreover, if the individual estimators $\covestij$ are sufficiently uncorrelated, then the overall variance estimator will achieve high precision.

We consider three different types of terms in the double sum.
The first type is when the covariance between unit-level estimators $\Cov{\dhtesti, \dhtestj}$ is known to be zero.
In particular, the linear response assumption implies that $\Cov{\dhtesti, \dhtestj} = 0$ when $\Cov{\dosei, \dosej} = 0$.
Therefore, we set the weighting function $R_{i,j}(\dosei, \dosej)$ to be exactly zero whenever $\Cov{\dosei, \dosej} = 0$, which means that $\covestij = 0$.

The second type is the non-zero, off-diagonal terms: $\Cov{\dhtesti, \dhtestj} \neq 0$ and $i \neq j$.
For these pairs, we define the overall weighting function $R_{i,j}(\dosei, \dosej)$ to be
\[
R_{i,j}(\dosei, \dosej)
= Q_{i,j}(\dosei, \dosej) - S_{i,j}(\dosei, \dosej)
   \enspace,
\]
where the two weighting functions $Q_{i,j}(\dosei, \dosej)$ and $S_{i,j}(\dosei, \dosej)$ are defined as
\begin{align}
    Q_{i,j}(\dosei, \dosej) &= \frac{\dosei - \Exp{\dosei}}{\Var{\dosei}} \cdot \frac{\dosej - \Exp{\dosej}}{\Var{\dosej}} ,\\
    S_{i,j}(\dosei, \dosej) &= a_{i,j} \paren{\dosei \dosej - \Exp{\dosei \dosej}}
    + b_{i,j} \paren{\dosei - \Exp{\dosei}}
    + c_{i,j} \paren{\dosej - \Exp{\dosej}} \enspace,
\end{align}
and the coefficients $a_{i,j}$, $b_{i,j}$, $c_{i,j}$ are obtained as solutions to the system of linear equations:

\newcommand{\varestlinearsystem}{
\[
\covmatij \begin{bmatrix}
a_{i,j} \\
b_{i,j} \\
c_{i,j}
\end{bmatrix}
=
\begin{bmatrix}
\Var{\dosei \dosej} & \Cov{\dosei, \dosei \dosej} & \Cov{\dosej, \dosei \dosej} \\
\Cov{\dosei \dosej, \dosei} & \Var{\dosei} & \Cov{\dosej, \dosei} \\
\Cov{\dosei \dosej, \dosej} & \Cov{\dosei, \dosej} & \Var{\dosej}
\end{bmatrix}
\begin{bmatrix}
a_{i,j} \\
b_{i,j} \\
c_{i,j}
\end{bmatrix}
=
\begin{bmatrix}
1 \\
0 \\
0
\end{bmatrix}
\enspace.
\]
}
\varestlinearsystem

The two weighting functions have been constructed such that $\poi \poj Q_{i,j}(\dosei, \dosej)$ and $\poi \poj S_{i,j}(\dosei, \dosej)$ are unbiased estimators of $\Exp{\dhtesti \dhtestj}$ and $\Exp{\dhtesti} \Exp{\dhtestj}$, respectively.
Because $\Cov{\dhtesti, \dhtestj} = \Exp{\dhtesti \dhtestj} - \Exp{\dhtesti} \Exp{\dhtestj}$, this means that the overall estimator $\covestij \triangleq \poi \poj R_{i,j}(\dosei, \dosej) = \poi \poj Q_{i,j}(\dosei, \dosej) - \poi \poj S_{i,j}(\dosei, \dosej)$ is an unbiased estimator of the covariance $\Cov{\dhtesti, \dhtestj}$.
We refer the reader to Appendix~\ref{sec:interval-proofs} for more details.

The third case is the diagonal terms: $i = j$.
In principle, we could use the same functions $Q_{i,j}(\dosei, \dosej)$ and $S_{i,j}(\dosei, \dosej)$ as for the off-diagonal terms, but the system of equations is underdetermined, so the coefficients $a_{i,j}$, $b_{i,j}$, and $c_{i,j}$ are not uniquely determined.
To address this, we set $c_{i,i} = 0$ when $i = j$, and obtain the coefficients $a_{i,i}$ and $b_{i,i}$ as solutions to the system of linear equations:
\[
\covmatii \begin{bmatrix}
a_{i,i} \\
b_{i,i}
\end{bmatrix}
=
\begin{bmatrix}
\Var{\dosei^2} & \Cov{\dosei^2, \dosei} \\
\Cov{\dosei^2, \dosei} & \Var{\dosei} \\
\end{bmatrix}
\begin{bmatrix}
a_{i,i} \\
b_{i,i}
\end{bmatrix}
=
\begin{bmatrix}
1 \\
0
\end{bmatrix}
\enspace.
\]

The existence of unique solutions to these systems of linear equations requires that the exposures and their products are not perfectly correlated.
The matrices $\covmatij$ allow us to capture this.
Note that $\covmatij$ is the 3-by-3 covariance matrix of the two exposures $\dosei$, $\dosej$, and their product $\dosei \dosej$ when $i \neq j$, and that it is the 2-by-2 covariance matrix of the exposure $\dosei$ and its square $\dosei^2$ when $i = j$.
For a given pair $i, j \in \outunits$ (either distinct or not), a unique solution to the corresponding system of linear equations above exists if $\det(\covmatij) > 0$.
To understand this, note that the determinant of a covariance matrix is a quantitative measure of the linear dependence of a set of random variables, which prompted \citet{Wilks1932Certain} to refer to $\det(\Sigma)$ as the ``generalized variance''.

In Appendix~\ref{sec:closed-form-coefficients}, we derive the coefficients in closed form.
The coefficients are there expressed as simple functions of various statistics of the joint distribution of exposure pairs.
The exposure distribution depends on the underlying bipartite graph and the experimental design, both of which are known to the experimenter.
Thus, the coefficients of the variance estimator can be computed before the experiment begins.
In the case of an independent cluster design, the coefficients may be computed exactly; for more complicated designs, the statistics in the coefficients may be estimated to high precision via a Monte Carlo procedure \citep{Fattorini2006}.
Throughout the paper, we assume that the exact coefficients are used.

\subsection{Unbiased and consistent variance estimation}

The following theorem demonstrates that, under certain condition on the exposure distribution described above, the variance estimator is unbiased.

\begin{assumption}[Non-degenerate Exposures]\label{assumption:non-degenerate-exposures}
For each pair of outcome units $i, j \in [n]$, the joint distribution of their exposures $\dosei, \dosej$ satisfies the non-degeneracy condition $\det(\covmatij) > 0$.
\end{assumption}

\newcommand{\varestthm}{
Under Assumption~\ref{assumption:non-degenerate-exposures} and the linear response assumption, the variance estimator of the \doseht point estimator is unbiased, i.e.
$
\Exp{\erlvarest} 
= \Var{\dhtest}
$.
}
\begin{thm}[Unbiased Variance Estimator] \label{thm:variance-estimate}
\varestthm
\end{thm}

As discussed above, the exposure distribution is known to the experimenter and so they can check whether the non-degeneracy conditions in Assumption~\ref{assumption:non-degenerate-exposures} hold before the experiment is run.
In particular, this information may inform the experimenter's choice of design.
In the case of a single exposure (i.e. $i = j$), the condition states that $\dosei$ and $\dosei^2$ are non-perfectly correlated random variables, which occurs exactly when $\dosei$ is supported on at least three distinct values.
Typically, it may be possible to ensure this condition holds through careful selection of the experimental design; however, such a condition cannot be satisfied when an outcome unit is incident to only a single diversion unit, as only two exposures are ever observed in that case.
This impossibility is in line with the fact that unbiased variance estimation is generally not possible in experimental settings with binary treatments.
In the case of a pair of exposures (i.e. $i \neq j$), the condition states that exposures $\dosei$, $\dosej$, and their product $\dosei \dosej$ are not perfectly correlated random variables.
Again, this may be achieved by careful selection of the experimental design; however, this condition cannot be satisfied when two outcome units have identically weighted edges (i.e. $w_{i,k} = w_{j,k}$ for all $k \in \divunits$), because this would imply that the exposures are identical, $\dosei = \dosej$.

When Assumption~\ref{assumption:non-degenerate-exposures} does not hold, our proposed variance estimator will be ill-defined or biased. 
Indeed, it is possible that no unbiased variance estimator exists in such settings.
In this case, one can replace the problematic $\Cov{\dhtesti, \dhtestj}$ terms, which cannot be estimated directly, with upper bounds that can be estimated.
For example, if the exposure $\dosei$ takes only two values so that $\det(\covmatii) = 0$, then one can replace the problematic term in the variance with the upper bound: $\Cov{\dhtesti, \dhtesti} =\Var{\dhtesti} = \Exp{\dhtesti^2} - \Exp{\dhtesti}^2 \leq \Exp{\dhtesti^2}$, as  \citet{Aronow2013Conservative} do when they invoke Young's inequality.
To unbiasedly estimate this term, we can modify our weighting estimator as $R_{i,i}(\dosei) = Q_{i,i}(\dosei)$.
Similarly, if $\det(\covmatij) = 0$ for some distinct outcome units $i \neq j$, then one can replace the corresponding covariance term with an upper bound obtained from Cauchy-Schwarz and AM--GM inequalities:
\[
  \Cov{\dhtesti, \dhtestj}
  \leq \sqrt{\Var{\dhtesti} \Var{\dhtestj}}
  \leq \frac{1}{2} \paren[\Big]{\Var{\dhtesti} + \Var{\dhtestj}} \enspace.
\]
An unbiased estimator for the above term may be obtained by modifying the weighting function $R_{i,j}(\dosei, \dosej) = 1/2 \cdot \paren{R_{i,i}(\dosei, \dosei) + R_{j,j}(\dosej, \dosej)}$.
Under Assumptions~\ref{assumption:bounded-potential-outcomes} and \ref{assumption:design}, replacing one of these individual terms in this way leads to a positive bias of the normalized variance estimator $n \cdot \erlvarest$ which is on the order $\bigO{1/n}$.
Thus, the variance estimator remains asymptotically unbiased as long as we apply these upper bounds to $\littleO{n}$ terms.
This is stated formally in the proposition below:

\begin{prop}\label{prop:conservative-var-estimator}
Let $\mathcal{E} = \setb{ (i,j) \in \outunits \times \outunits : \det(\covmatii) = 0}$ be the pairs of outcome units for which Assumption~\ref{assumption:non-degenerate-exposures} is not satisfied.
Consider the alternative variance estimator $\erlvarest'$ defined by the new weighting function:
\[
R'_{i,j}(\dosei, \dosej)
= \left\{
     \begin{array}{ll}
     R_{i,j}(\dosei, \dosej) &\text{if } (i,j) \notin \mathcal{E} \\
    Q_{i,j}(\dosei, \dosei) &\text{if } (i,j) \in \mathcal{E} \text{ and } i = j\\
     \frac{1}{2} \paren{R'_{i,i}(\dosei, \dosei) + R'_{j,j}(\dosei, \dosej)} &\text{if } (i,j) \in \mathcal{E} \text{ and } i \neq j\\
     \end{array}
  \right.
  \enspace.
\]
Then, under Assumptions~\ref{assumption:bounded-potential-outcomes} and \ref{assumption:design}, then normalized alternative variance estimator is conservative in expectation with bounded bias: $0 \leq \Exp{n \cdot \erlvarest'} - n \cdot \Var{\dhtest} = \bigO{\abs{\mathcal{E}} / n}$.
Thus, the normalized alternative variance estimator is asymptotically unbiased if $\abs{\mathcal{E}} = \littleO{n}$.
\end{prop}

Even when the terms can be estimated without bias, it could still be preferable to use the bound here if $\det(\covmatij) \approx 0$, as the individual covariance estimators would have high variance in this case, causing the overall variance estimator to be imprecise.
For the rest of the section, we analyze the variance estimator assuming that Assumption~\ref{assumption:non-degenerate-exposures} holds exactly, but we conjecture that many of our results will go through for the conservative variance estimator described above.

We now present conditions under which our proposed variance estimator is consistent in mean squared error.
To this end, we define $\Delta = \min_{i,j \in [n]} \det(\covmatij)$ to be the smallest non-degeneracy measure amongst all distinct pairs ($i \neq j$) of exposures and single exposures ($i=j$).

\newcommand{\consistentvarestimation}{
Under Assumptions~\ref{assumption:bounded-potential-outcomes}, \ref{assumption:design}, and \ref{assumption:non-degenerate-exposures}, the mean squared error of the normalized variance estimator is bounded as 
\[
\Exp[\Big]{\paren[\big]{n \cdot \Var{\dhtest} - n \cdot \erlvarest}^2} 
= \bigO[\Big]{ \frac{1}{n} \cdot \paren[\Big]{\divdeg^3 \outdeg^7 + \frac{1}{\Delta^2}}}
\enspace.
\]
Thus, the normalized variance estimator is consistent if $\divdeg^3 \outdeg^7 = \littleO{n}$ and $\Delta = \omega \paren{n^{-1/2}}$.
}
\begin{thm} \label{thm:consistent-var-estimation}
\consistentvarestimation
\end{thm}

In Theorem~\ref{thm:consistent-var-estimation}, we analyze convergence of the normalized variance estimator to the normalized variance, which is bounded below by a positive constant, according to Assumption~\ref{assumption:mse-rate}.
This normalization ensures that both the variance and its estimator are on appropriate scales so that the mean squared error does not trivially approach zero.
In light of Assumption~\ref{assumption:mse-rate}, this says that the variance estimator converges to the variance at a faster rate than the variance of the \doseht estimator converges to zero.
A stronger requirement for the rate of convergence yields a stronger restriction on the growth conditions of the bipartite graph.
We also require that the degeneracy measure $\Delta$ does not approach zero too quickly, which is a quantitative strengthening of Assumption~\ref{assumption:non-degenerate-exposures}.

Theorem~\ref{thm:consistent-var-estimation} holds for a broad class of designs under Assumption~\ref{assumption:design}, including Bernoulli randomization at the level of individual diversion units which ignores the structure of the bipartite graph.
An experimental design which takes the structure of the bipartite graph into account would generally yield consistency under weaker conditions than those needed for Theorem~\ref{thm:consistent-var-estimation}.
In particular, a design which decorrelates the individual covariance estimators $\covestij$ and make their variances small will yield improved precision of the variance estimator.
Interestingly, improving the precision of the \doseht estimator and improving the precision of its variance estimator are generally different goals: a design which improves one may not necessarily improve the other.
A detailed investigation into this trade-off is beyond the scope of the current paper.

\subsection{Asymptotically valid confidence intervals}

We may now use our variance estimator together with the asymptotic normality to construct well-motivated confidence intervals.
We may construct $1 - \alpha$ confidence intervals by
\[
\dhtest \pm \Phi^{-1}(1 - \alpha/2) \sqrt{\erlvarest} \enspace,
\]
where $\Phi^{-1}:[0,1] \rightarrow \Reals $ is the quantile function of the standard normal deviate.
The following result follows from the consistency of the variance estimator (Theorem~\ref{thm:consistent-var-estimation}) together with asymptotic normality of the \doseht estimator (Theorem~\ref{thm:asymptotic-normality}).

\newcommand{\asymptoticvalidity}{
Under Assumptions~\ref{assumption:bounded-potential-outcomes}-\ref{assumption:non-degenerate-exposures} and further supposing that $\outdeg^4 \divdeg^{10} = \littleO{n}$ and $\Delta = \omega \paren{n^{-1/2}}$, a Wald-type confidence interval using the proposed variance estimator is asymptotically valid:
\[
\lim_{n \rightarrow \infty}
\Pr \paren[\Big]{ \ate \in \bracket[\Big]{ \dhtest \pm \Phi^{-1}(1 - \alpha/2) \sqrt{\erlvarest}}  } = 1 - \alpha \enspace.
\]
}
\begin{corollary}\label{cor:asymptotic-validity}
\asymptoticvalidity
\end{corollary}

In practice, it is possible that the proposed variance estimator may take negative values; in particular, this may happen when the variance is near zero and the variance estimator is imprecise relative to the variance.
When the variance estimator takes a negative value, this construction of confidence intervals is not well-defined.
We suggest two possible alternatives here: the experimenter may either use the absolute value of the variance estimator under the square root, or the experimenter may use a more conservative variance estimate.

\section{Analyzing \doseht Without the Linear Response Assumption} \label{sec:no-assumption}

Our previous analysis of the \doseht estimator relied on the linear response assumption.
In this section, we show that without the linear response assumption, the \doseht estimator can be interpreted as capturing as an average of linear approximations of each unit's dose response function to treatment intensities among the relevant diversion units.
The following theorem derives the expectation of the \doseht estimator without the linear response assumption.

\newcommand{\expectationnoassumption}{
Let the potential outcome functions be arbitrary functions of the exposures: $\poi(\zv) = \poi(\dosei)$.
Then, the expectation of the \doseht estimator is 
\[
\Exp{\dhtest} 
= \frac{1}{n} \sum_{i=1}^n \tilde{\beta}_i 
\enspace,
\]
where $\tilde{\beta}_i$ is the coefficient of the exposure $\dosei$ in an OLS regression of $\poi$ on $\dosei$: $\tilde{\beta}_i = \frac{\Cov{ \dosei, \poi }}{\Var{\dosei}} $.
}
\begin{thm}[Arbitrary Responses] \label{thm:expectation-no-assumption}
\expectationnoassumption
\end{thm}

Theorem~\ref{thm:expectation-no-assumption} shows that under a general (non-linear) response assumption, the \doseht estimator may be interpreted as estimating the average of the slopes of the best linear fit of the outcome to the exposure.
We emphasize that this regression cannot be run by the experimenter because the outcomes are not known.

This result is related to several previous results within and outside causal inference.
Realizing that most conditional expectation functions are not linear, some statisticians and econometricians have advocated for an interpretation of linear regression as capturing an interpretable approximation of the underlying relationship between the outcome and the regressors \citep{Chamberlain1984Panel,Manski1991Regression,Goldberger1991Course}.
Specifically for causal inference, \citet{Angrist1998Estimating} highlights that when linear regression is used to estimate treatment effects in an observational setting, the estimator captures a variance-weighted average of unit-level causal effects (see also \citealp{Aronow2015Does} and \citealp{Sloczynski2020Interpreting}).
In a similar vein to these results, Theorem~\ref{thm:expectation-no-assumption} shows that the \doseht estimator captures a policy-relevant causal quantity even if the linear response assumption does not hold.
The difference is that the effect it captures is an unweighted average over the units, and the approximation is with respect to each unit's response function.

Under the linear response assumption, this regression-based estimand is equal to the average total treatment effect (ATTE) defined in Section~\ref{sec:causal-estimand}.
However, these two estimands will not coincide for arbitrary response functions and designs.
Aside from the linear response assumption, there are several scenarios where we will expect the ATTE and the regression-based estimand to be similar.
The first scenario is when the design very closely approximates the Bernoulli design so that exposures have mean $1/2$ and concentrate around $0$ and $1$.
When the design is exactly Bernoulli, one can verify that the regression-based estimand is exactly equal to the ATTE, which matches the  intuition from the no-interference setting.
The second scenario is when the response function is well-approximated by a linear function.
An extensive investigation into formal conditions under which the regression-based estimand and the causal estimand (ATTE) are equivalent is beyond the scope of the current paper.

\section{A Cluster Design for Targeting Exposure Distribution}\label{sec:design}

In this section, we describe \exposuredesign, an independent cluster design which aims to improve precision of the \doseht estimator by constructing a desirable exposure distribution.
To this end, we first show in Section~\ref{sec:precision} that increasing the variance of exposures and decreasing the covariance between exposures can lead to improved precision of the \doseht estimator in settings of interest.
In Section~\ref{sec:clustering-objective}, we present a clustering objective which aims to achieve such exposure distributions, thereby improving the precision of \doseht estimator.
Finally, we present a heuristic algorithm for optimizing this clustering objective in Section~\ref{sec:corr-clustering-implementation}.

\subsection{An ideal exposure distribution}\label{sec:precision}

Like all linear estimators, the \doseht estimator will incur a large mean squared error when the coefficients for the observed outcomes are large.
In particular, if the variance of an exposure $\Var{\dosei}$ is close to zero, the corresponding term of the estimator in \eqref{eq:ht-exposure-def} will become large, yielding a high mean squared error even though the estimator is unbiased.
In general, experimenters should use designs for which the corresponding exposure variances are large.

However, large exposure variances should not be the only property of the exposure distribution that experimenters focus on.
Consider a naive design that places equal probability on two treatment vectors: either all diversion units receive treatment ($\zv = \onevec$) or all diversion units receive control ($\zv = \zerovec$).
Under this design, all of the exposure variances are 1, which is the largest possible variance.
However, we observe either all of the treatment outcomes or all of the control outcomes, but never a mix of the two; in fact, the estimator itself takes only two values.
Thus, the \doseht estimator will suffer very large MSE under this design, despite the individual exposure variances being as large as possible.
This raises the question: how should we construct a design that improves the precision of the \doseht estimator?

This is a challenging task, since the precision of the \doseht estimator depends on the unobserved outcomes.
Indeed, a universally optimal design does not exist \citep{Harshaw2021Balancing}.
However, we argue that a good heuristic is to construct the design so that the variance of the exposures are large and the covariance between most pairs of exposures are close to zero.
As discussed at the end of Section~\ref{sec:basic-stats}, a design which directly targets these aspects of the exposure distribution may hope to ensure high precision of the \doseht estimator under weaker growth conditions on the bipartite graph than those presented in our analysis.

As an illustration to motivate this heuristic, consider the scenario where all of the individual treatment effects are zero, i.e. the response functions are of the form $\poi(\dosei) = \incepti$.
Studies of these sort are sometimes called uniformity trails or A/A tests.
In this scenario, one may derive the MSE of the \doseht estimator as
\[
\Exp{\paren{ \dhtest - \ate}^2 } = 
\frac{1}{n^2} \bracket[\Bigg]{
    \sum_{i=1}^m \incepti^2 \frac{1}{\Var{\dosei}} 
    + 2 \sum_{i < j} \incepti \inceptj \frac{\Cov{\dosei, \dosej} }{\Var{\dosei} \Var{\dosej}}
} \enspace.
\]

As the individual variance terms increase, the first sum decreases.
The effect of the second term depends on the sign of the product of intercepts, $\incepti \inceptj$.
Generally speaking, these intercepts are unknown to the experimenter.
For the sake of this discussion, consider when the outcomes $\poi(\zv)$ are non-negative, in which case all intercepts $\incepti$ and their products $\incepti \inceptj$ also are non-negative.
In this case, decreasing the correlation between exposures would decrease the second term, leading to an overall decrease in the MSE of the \doseht estimator.

\subsection{Clustering objective for targeting exposure distribution}\label{sec:clustering-objective}

In the previous section, we noted that a reasonable heuristic is to assign treatments to the diversion units so that the variance of exposures is large and the covariance between most exposures is small.
However, as argued in Section~\ref{sec:linear-exposure-response}, constructing a treatment distribution which realizes a desired exposure distribution is generally not possible due to overlapping structures in the bipartite graph.
In this section, we present an optimization formulation for an independent cluster design which aims to achieve large exposure variance and small correlations between exposures, to the extent that this is possible given the bipartite graph.

We propose choosing a cluster design which maximizes the following objective function:
\begin{equation}\label{eq:cluster_objective}\tag{\exposuredesign}
\max_{\text{clustering } \clustering}
 \quad \sum_{i=1}^n \bracket[\bigg]{\Var{\dosei} - \phi \sum_{i \neq j} \Cov{ \dosei, \dosej}}
\end{equation}
The variance and covariance of the exposures above are with respect to the random treatment assignments of the corresponding independent cluster design.
The first term in the objective is the sum of the exposure variances, so maximizing this term will encourage large exposure variances.
The second term penalizes positive correlation between exposures, and maximizing it encourages small correlation.
The correlation penalizing parameter $\phi \geq 0$ controls the relative emphasis between large exposure variances and small exposure correlations.
When $\phi = 0$, then the emphasis is placed entirely on increasing individual exposure variance; this is typically undesirable, as the optimal solution is often a single cluster containing all diversion units, which results in a design where either all diversion units receive treatment or all diversion units receive control.
Increasing $\phi$ places more emphasis on decorrelating exposures.

A key insight to solving the \exposuredesign formulation is that it may be reformulated as a \emph{correlation clustering} problem, which is well-studied in the algorithms literature \citep{BBC02correlation,  Swamy04correlation, CGV05clustering}.
The existing computational understanding of these correlation clustering problems is another reason to use the \exposuredesign objective.
The following proposition states the re-formulation of the \exposuredesign objective into the correlation clustering variant, denoted \corrclust.

\begin{prop}\label{prop:corr_cluster_instance}
For each pair of diversion units $i,j \in \divunits$, define the value $\omega_{i,j} \in \Reals$ as
\begin{equation} \label{eq:corr-clustering-weights}
\omega_{i,j} 
= \paren{1 + \phi} \sum_{k=1}^m w_{k,i} w_{k,j}
- \phi \paren[\Big]{\sum_{k=1}^m w_{k,i}} \paren[\Big]{\sum_{k=1}^m w_{k,j}} \enspace,
\end{equation}
where $w_{k,i}$ is the weight of the edge between the $k$th outcome unit and the $i$th diversion unit.
\ref{eq:cluster_objective} is equivalent to the following clustering problem:
\begin{equation}\label{eq:corr_clustering_obj}\tag{\corrclust}
\max_{\text{clusterings } \clustering }
\quad \sum_{\cluste{r} \in \clustering} \sum_{i, j \in \cluste{r}} \omega_{i,j} 
\enspace.
\end{equation}
\end{prop}

Although \corrclust is a variant of the weighted maximization-type correlation clustering problems previously studied in the literature \citep{CGV05clustering, Swamy04correlation}, 
it is not equivalent to previously studied formulations in an approximation-preserving sense, as it takes positive and negative values.
Given that weighted maximization correlation clustering is NP-Hard \citep{CGV05clustering}, it is reasonable to presume that our formulation \corrclust is also computationally hard.
However, these computational complexity considerations are beyond the scope of this paper.

We remark that \exposuredesign places no explicit constraint on the number of clusters produced by the clustering algorithm.
However, our analysis in Section~\ref{sec:basic-stats} suggests that limiting the cluster sizes, and thereby reducing correlation between exposures, helps to achieve consistency and normality of the \doseht estimator.
This desirable cluster structure is not captured by the optimization problem itself, but we handle it through our local search heuristic described in Section~\ref{sec:corr-clustering-implementation}.
We also remark that the \exposuredesign objective does not directly minimize MSE of the \doseht estimator, but should instead be understood as a useful heuristic.
Finally, we remark that it is impossible to induce a negative correlation between exposures in the class of independent cluster designs, so the second term of the objective attains its maximum when the exposures are uncorrelated.

The \exposuredesign is conceptually similar to the correlation-clustering based design presented in \cite{pouget2019variance}, but it differs in several key ways.
\exposuredesign provides experimenters the flexibility to trade-off larger exposure variances with more de-correlated exposures by setting the parameter $\phi$.
In contrast, the cluster design of \cite{pouget2019variance} focuses solely on the exposure variance by maximizing which is referred to as ``empirical dose variance’’ in their paper. 
As we demonstrate in Appendix~\ref{sec:design-proofs}, their objective is equal to ours when the trade-off parameter is set to $\phi = 1/(n-1)$.
In this sense, their cluster design can be viewed as a specific instance of the more general \exposuredesign, where a greater emphasis is placed on maximizing the exposure variances.
More importantly, the \exposuredesign presented in this paper is designed to increase the precision of the \doseht estimator, while the correlation-clustering based design of \cite{pouget2019variance} is motivated by the intuition that extreme exposures are helpful in this setting.

\subsection{Local search heuristic for \exposuredesign} \label{sec:corr-clustering-implementation}
We now describe a local search heuristic for optimizing \exposuredesign.
The local search is initialized with the singleton clustering and iteratively seeks to improve the clustering.
In each iteration, the algorithm loops through random pairs of diversion units $i,j \in \divunits$ and moves diversion unit $j$ to the cluster currently containing diversion unit $i$ if that change improves the objective value, subject to a user-defined constraint on the clusters.
The local search algorithm is presented more formally below as Algorithm~\ref{alg:local-search}.

\begin{algorithm}[H]
    \caption{Local Search($\weightM, \phi, \maxcorr, T$, cluster constraints)} \label{alg:local-search}
    Initialize singleton clustering $\clustering = \setb{ \setb{1}, \setb{2}, \dots \setb{m} }$ \\
    \For{iterations $t = 1 \dots T$}{
        Choose a uniformly random permutation $\pi$ on the diversion units. \\
        \For{diversion units $i \in \pi$}{
            Randomly select a diversion unit $j$ with probability proportional to $\paren{\weightM^\tran \weightM}_{i,j}$ \\ 
            Let $\cluste{}$ and $\cluste{}'$ be the clusters containing diversion units $i$ and $j$, respectively. \\
            \If{moving $j$ from cluster $\cluste{}'$ to $\cluste{}$ increases objective value \& satisfies user-defined constraints}{
                Move diversion unit $j$ from cluster $\cluste{}'$ to $\cluste{}$. \\
            }
        }
    }
    \Return{clustering $\clustering$}
\end{algorithm}

Given a diversion unit $i$, we use \emph{wedge sampling} to select unit $j$ proportional to $\paren{\weightM^\tran \weightM}_{i,j} = \sum_{k=1}^n w_{k,i} w_{k,j}$ \citep{CL99approximating}.
We use wedge sampling because picking pairs of units for which $\paren{\weightM^\tran \weightM}_{i,j}$ is large often results in a large correlation clustering weight $\omega_{i,j}$.
Performance improvements are obtained by computing the correlation clustering weights $\omega_{i,j}$ only when they are needed to evaluate changes in the objective.
In particular, the first term of \eqref{eq:corr-clustering-weights} is an inner product whose computation scales with the sparsity of the bipartite graph and the second term is the product of sums which may be pre-computed.

Diversion unit $j$ is moved into the cluster containing diversion unit $i$ if two conditions are met: the objective increase and the user-defined cluster constraints are satisfied.
We recommend that experimenters choose constraints which limit the cluster sizes in some way.
For example, the experimenter may choose to constraint the number of diversion units within a cluster.
In our implementation, we constrain the sum of the (unweighted) degrees of diversion units within a cluster to be a fixed fraction of the total number of edges.
In this way, no cluster has too many outgoing edges to outcome units.
Constraining the clusters in this way implicitly limits the amount of dependence between exposures, which is one of the key aspects underlying the design conditions in Assumption~\ref{assumption:design} of our analysis.

This local search algorithm is different from the one presented in \cite{pouget2019variance}, which approximates the Gram matrix $\weightM^\tran \weightM$ offline as the sum of a sparse matrix and a rank-one matrix, so that the algorithm works with an approximation to the objective.
In contrast, our algorithm accepts and rejects changes based on the exact value of the objective.
Relative to \cite{ES09bounding}, this local search does not consider moving units to new empty clusters, nor does it consider merging clusters. 
Moves of the first type seem consistently unprofitable in our setting. 
As for merges, we find that the algorithm is able to essentially perform them by moving one diversion unit at a time.

\section{An Application to Online Marketplace Experiments}\label{sec:simulations}

In this section, we apply our proposed methodology to a simulated marketplace experiment based on a product review dataset from the Amazon marketplace \citep{mcauley2015image, he2016ups}.
The Amazon product review dataset contains 83 million reviews made by 121 thousand customers on 9.8 million items.
In this application, we imagine running an experiment where we change the pricing mechanism of items in the marketplace, and are interested in how a customer's reported satisfaction is affected by this change in pricing mechanism.
The items sold in the marketplace are the diversion units and the customers in the marketplace are the outcome units.
An edge is present in the bipartite graph if a customer reviewed an item and all edges incident to an outcome unit are uniformly weighted.
Thus, a customer's exposure is the unweighted average of the treatment status of the items they have previously reviewed.

In our simulated marketplace experiment, we generate potential outcomes via an exposure-response function.
The outcomes themselves are the satisfaction score of a customer given their exposure; the responses in this study are simulated, but we can imagine that they are either reported directly by a customer or constructed based upon text analysis of the customer's review.
In the case of a linear response, a positive slope indicates an increase in customer satisfaction as a result of the new pricing mechanism, while a negative slope indicates a decrease in satisfaction as a result of the new pricing mechanism.

We preprocess the Amazon produce review dataset for computational tractability in the same manner as \cite{pouget2019variance}.
We begin by removing customers that have reviewed fewer than 100 items. 
Next, we execute a balanced partitioning algorithm \citep{aydin2019distributed} on the entire bipartite graph to create groups of customers and groups of items.
After this preprocessing, we define the diversion units to be the item groups and the outcome units to be the customer groups.
The resulting bipartite graph has 1 thousand outcome units, 2.4 million diversion units, and 7.1 million edges.
We emphasize that this bipartite graph does not appear to satisfy the growth conditions (specified in Section~\ref{sec:basic-stats}) required for consistency and normality under the broad class of designs captured by Assumption~\ref{assumption:design}.
In this sense, this application may be considered a test of the efficacy of the proposed \exposuredesign under weaker growth conditions.

We investigate the statistical properties of the \doseht estimator, the variance estimator, and the resulting confidence intervals under various treatment designs in this application.
In particular, we compare our proposed \exposuredesign to several existing designs: the Bernoulli design, the correlation clustering design of \citet{pouget2019variance}, and the balanced partitioning cluster design of \cite{eckles2016design}, as implemented by \citet{aydin2019distributed}.
Although the balanced partitioning design was not developed for the bipartite setting, we may expect it to achieve high precision estimates if the clustering produces decorrelated exposures with large variances.

We generate the potential outcomes in three simulations, where we vary the response functions that are used.
The first two simulations feature linear response functions and the third simulation features a non-linear response function.
We remark that although the parameters of the response functions are randomly chosen in our simulations, this random parameter draw is made only once and the outcomes themselves are fixed across all sampled assignments of all designs.
These simulations are listed below.

\begin{itemize}
    \item 
    \textbf{S1: (Mostly) Positive Treatment Effect.}
    In this simulation, we set almost all of the individual treatment effects to be positive across units, while varying the responses amongst the units. 
    More precisely, we sample the slope terms as $\slopei \sim \mathcal{N}(2, 1)$ and the intercept terms as $\incepti \sim \mathcal{N}(-1, 3/8)$.
    \item
    \textbf{S2: (Nearly) Zero Treatment Effect.}
    In this simulation, we set all the individual treatment effects close to zero, while varying the baseline outcomes.
    The outcomes are chosen to be mostly positive.
    More precisely, we sample the slope terms as $\slopei \sim \mathcal{N}(0, 1/2)$ and the intercept terms as $\incepti \sim \mathcal{N}(2, 3/8)$.
    \item 
    \textbf{S3: Non-Linear Response.}
    In this simulation, we use a non-linear response function to specify the potential outcomes.
    In particular, the response of outcome unit $i$ is $\poi(\dosei) = 4 \dosei (\dosei - 1) + \incepti$, where $\incepti \sim \mathcal{N}(0, 1/8)$.
    Under this response, all individual treatment effects are 0.
    Because the linear response assumption is not satisfied, we should not expect our statistical analysis (unbiasedness, consistency, normality, etc) to hold exactly.
\end{itemize}

We run \exposuredesign with different values of the correlation penalty parameter $\phi$, chosen from a grid of 10 points between $[0,2]$.
The clustering itself is obtained using our local search heuristic presented in Section~\ref{sec:corr-clustering-implementation}.
Recall that the correlation clustering objective of \citet{pouget2019variance} may be obtained by setting $\phi = 1/(n-1)$.
For this reason, we compute the corresponding cluster by running our local search heuristic with $\phi = 0.001 \approx 1/(n-1)$.

\begin{table}[h]
    \centering
	\begin{tabular}{clccccc}
		\multicolumn{1}{l}{} &  & \begin{tabular}[c]{@{}c@{}}Exposure\\ Design \\ ($\phi = 0.223$) \end{tabular} &
		\begin{tabular}[c]{@{}c@{}}Exposure\\ Design \\ ($\phi = 1.0$) \end{tabular} &
		\begin{tabular}[c]{@{}c@{}}Correlation \\ Clustering\end{tabular} & \begin{tabular}[c]{@{}c@{}}Balanced \\ Partitioning\end{tabular} & Bernoulli \\ \cline{2-7} 
		\multirow{3}{*}{\textbf{S1}} 
		& RMSE & 0.049 & 0.057 & 0.087 & 0.0718 & 0.659 \\
		& CI Width & 0.220 & 0.239 & 0.329 & 0.284 & 2.576 \\
		& CI Coverage & 91.5\% & 91.5\% & 94.0\% & 94.1\% & 95.1\% \\ \cline{2-7} 
		\multirow{3}{*}{\textbf{S2}} 
		& RMSE & 1.81 & 2.24 & 2.05 & 1.86 & 43.83 \\
		& CI Width & 7.04 & 8.72 & 8.00 & 7.21 & 190.8 \\
		& CI Coverage & 94.8\% & 94.8\% & 95.0\% & 94.5\% & 95.1\% \\ \cline{2-7} 
		\multirow{3}{*}{\textbf{S3}} 
		& RMSE & 0.86 & 1.09 & 0.90 & 0.78 & 24.37 \\
		& CI Width & 3.47 & 4.35 & 3.82 & 3.39 & 95.15 \\
		& CI Coverage & 95.7\% & 95.5\% & 96.9\% & 96.9\% & 95.1\% \\ \cline{2-7} 
	\end{tabular}
	\caption{Simulation results}
	\label{table:summary-stats}
\end{table}

A summary of the main results from these simulations appears in Table~\ref{table:summary-stats}.
For each treatment design and simulation, we sample $15,000$ exposure vectors, compute the observed outcomes, and construct the corresponding \doseht and variance estimators.
Given the \doseht and variance estimators, we construct the confidence intervals as described in Section~\ref{sec:confidence-intervals}, with absolute value corrections when the variance estimator takes a negative value.
For each simulation and treatment design, we report the root mean square error (RMSE) of the \doseht estimator, the average width of the $95\%$ confidence intervals, and the coverage of the $95\%$ confidence intervals.

In Table~\ref{table:summary-stats}, we show results for \exposuredesign with parameters $\phi = 0.223$ and $\phi = 1.0$.
The parameter value $\phi = 1.0$ is chosen arbitrarily, while the value $\phi = 0.223$ is the value for which \exposuredesign achieves the smallest mean squared error across all simulations.
We emphasize that selecting $\phi$ in this way (i.e. a course grid search to find the smallest MSE) cannot be performed by an experimenter and is presented here only to inform our discussion.
Still, it is interesting that a single penalty parameter minimized MSE for all simulations considered here.
See Figure~\ref{fig:mse-changing-phi} for comparison of the mean squared error of \exposuredesign when the correlation penalty parameter $\phi$ is varied.

We draw particular attention to a few features in these results.
\exposuredesign achieves the smallest RMSE in the simulations which satisfy the linear response assumption.
All cluster-based designs achieve significantly smaller RMSE than the Bernoulli design, which emphasizes the importance of carefully considering the exposure distribution when the growth conditions (specified in Section~\ref{sec:basic-stats}) are not satisfied.
The confidence intervals in Simulation 1 under \exposuredesign cover below the nominal 95\% level, indicating that either the sampling distribution of the \doseht estimator isn't sufficiently approximated by a normal or the variance estimator isn't sufficiently concentrated at this sample size.
The confidence intervals in Simulation 3 cover slightly above the nominal 95\% level, which is a result of conservative bias in the variance estimate due to non-linearity of the response. 

\begin{figure}[h!tb]
    \centering
\begin{subfigure}{0.25\textwidth}
\centering
  Positive Treatment Effect
\end{subfigure}\hfil
\begin{subfigure}{0.25\textwidth}
\centering
  Zero Treatment Effect
\end{subfigure}\hfil 
\begin{subfigure}{0.25\textwidth}
\centering
  Non-Linear Response
\end{subfigure}
\bigskip

\begin{subfigure}{0.25\textwidth}
  \includegraphics[width=\linewidth]{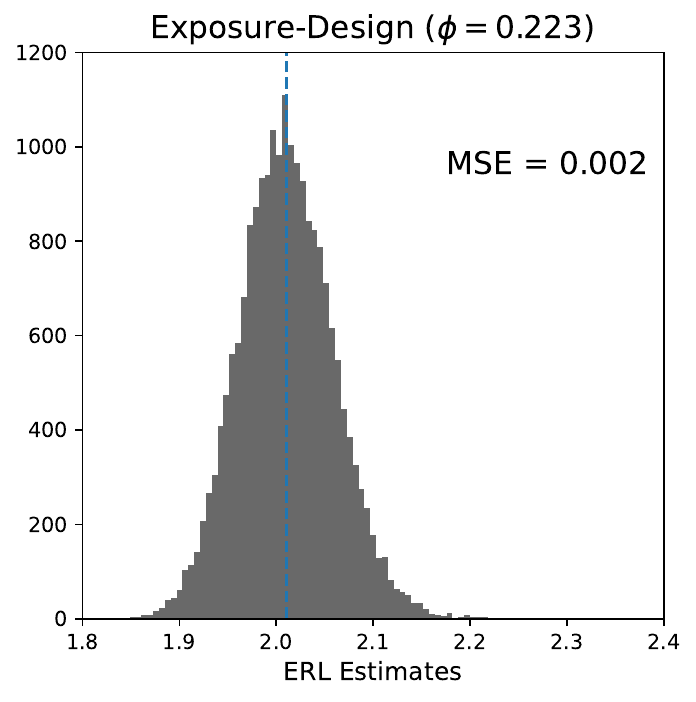}
  \label{fig:sim1_d1}
\end{subfigure}\hfil
\begin{subfigure}{0.25\textwidth}
  \includegraphics[width=\linewidth]{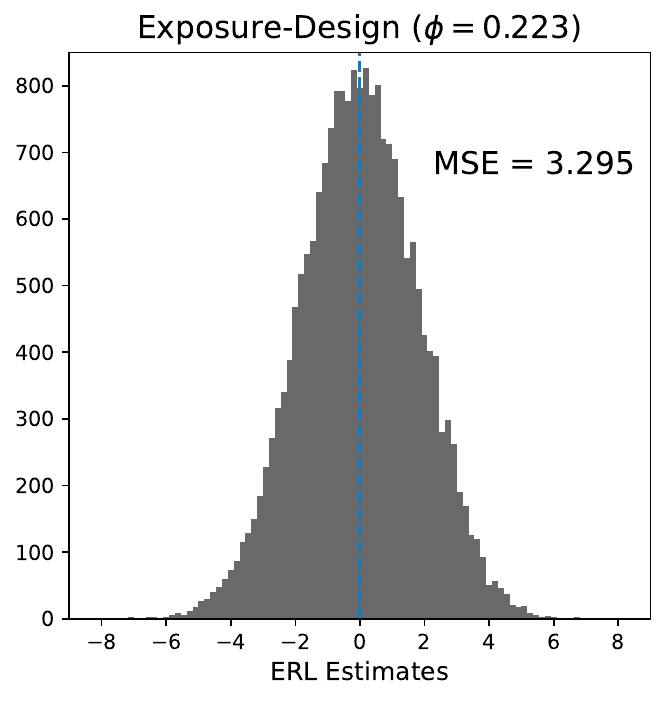}
  \label{fig:sim2_d1}
\end{subfigure}\hfil 
\begin{subfigure}{0.25\textwidth}
  \includegraphics[width=\linewidth]{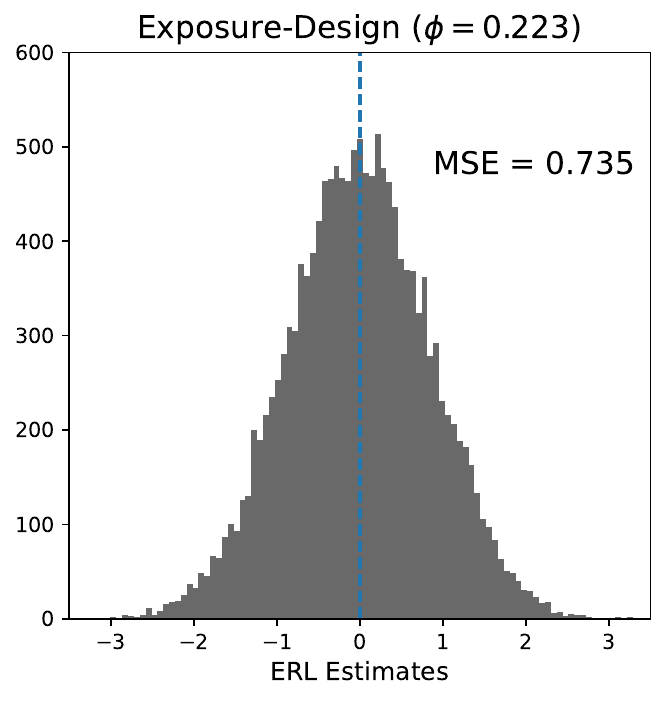}
  \label{fig:sim3_d1}
\end{subfigure}

\vspace{-.3cm}
\begin{subfigure}{0.25\textwidth}
  \includegraphics[width=\linewidth]{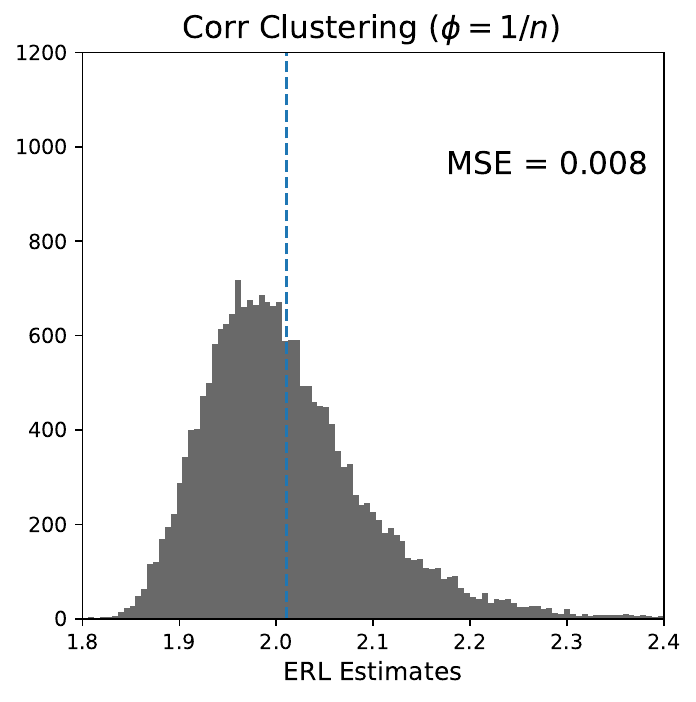}
  \label{fig:sim1_d2}
\end{subfigure}\hfil 
\begin{subfigure}{0.25\textwidth}
  \includegraphics[width=\linewidth]{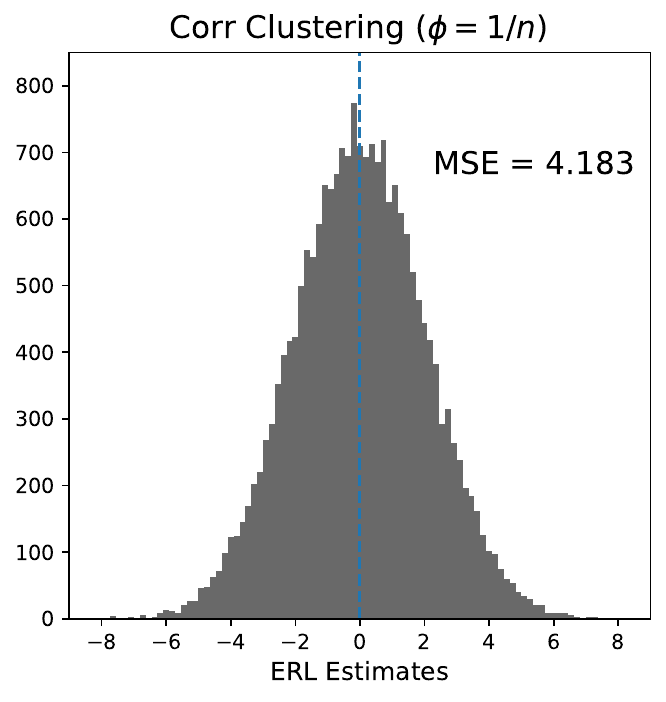}
  \label{fig:sim2_d2}
\end{subfigure}\hfil
\begin{subfigure}{0.25\textwidth}
  \includegraphics[width=\linewidth]{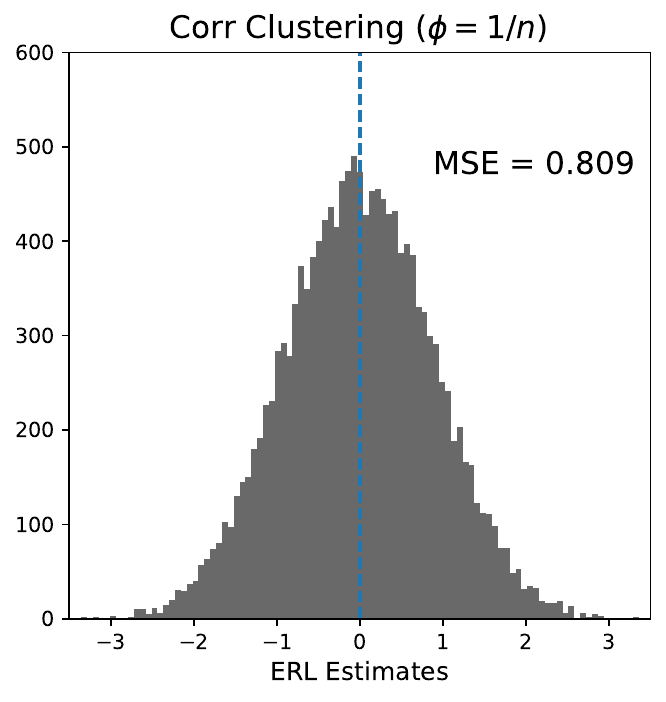}
  \label{fig:sim3_d2}
\end{subfigure}

\vspace{-.3cm}
\begin{subfigure}{0.25\textwidth}
  \includegraphics[width=\linewidth]{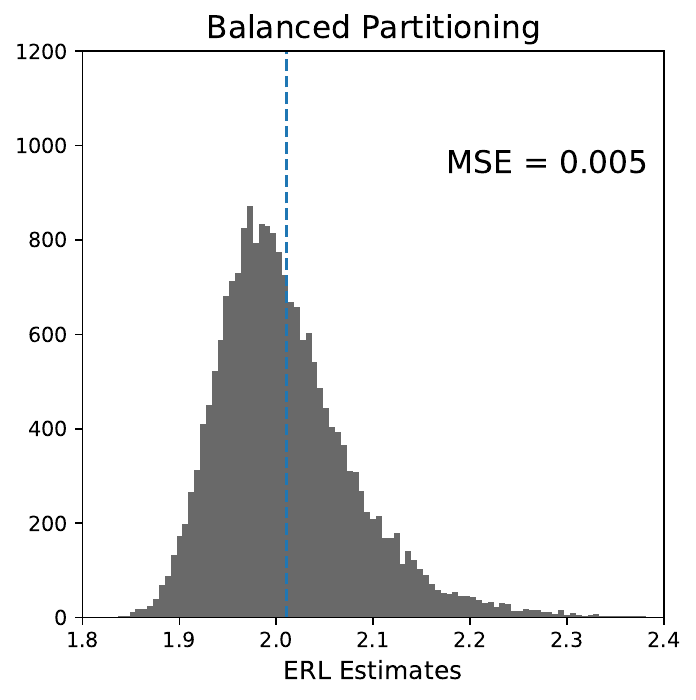}
  \label{fig:sim1_d3}
\end{subfigure}\hfil 
\begin{subfigure}{0.25\textwidth}
  \includegraphics[width=\linewidth]{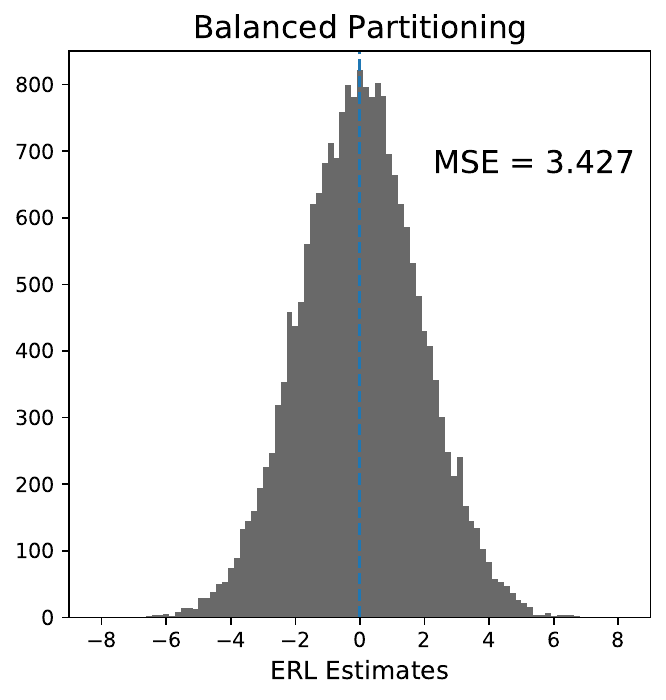}
  \label{fig:sim2_d3}
\end{subfigure}\hfil
\begin{subfigure}{0.25\textwidth}
  \includegraphics[width=\linewidth]{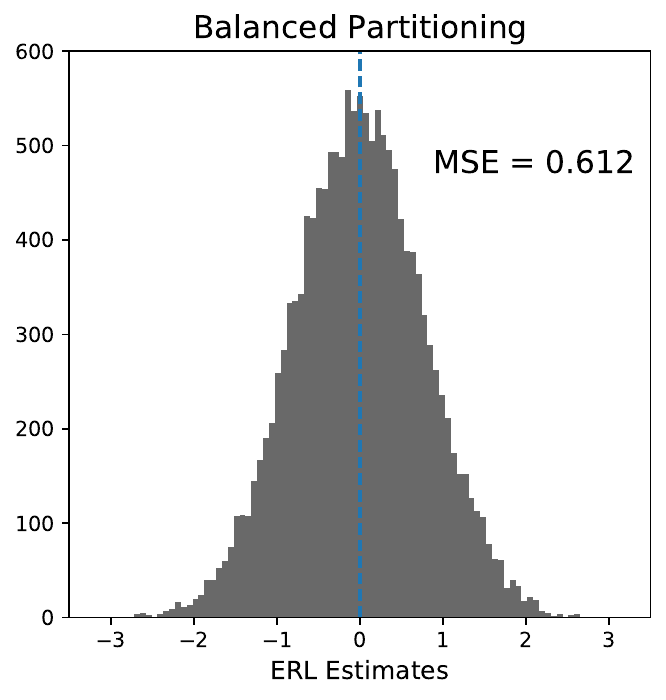}
  \label{fig:sim3_d3}
\end{subfigure}
\caption{Histograms of the \doseht estimator in simulations.}
\label{fig:erl-histograms}
\end{figure}

Figure~\ref{fig:erl-histograms} contains histograms of the \doseht estimator for each simulation and design, where the rows correspond to the designs and the columns correspond to the simulations.
The dotted vertical line in the plot is the true ATTE.
In all simulations, the distribution of the \doseht estimator appears unimodal, centered around the ATTE, and (roughly) normal, which is empirical evidence that the normal approximation used to derive confidence intervals may be well-motivated for \exposuredesign and other cluster-based designs.
This is to be expected for Simulations~1 and 2 where the linear response assumption holds, but is perhaps surprising in Simulation 3, which features a highly non-linear response.
This unbiasedness may be explained by Theorem~\ref{thm:expectation-no-assumption} in the following way: the quadratic responses in Simulation~3 yield zero treatment effect for all units.
Although the best linear approximation to each quadratic response does not well-approximate the quadratic response itself, the linear approximation has zero slope and so, in this sense, captures the ITE exactly.

\begin{figure}[h!tb]
    \centering
\begin{subfigure}{0.25\textwidth}
\centering
  Positive Treatment Effect
\end{subfigure}\hfil
\begin{subfigure}{0.25\textwidth}
\centering
  Zero Treatment Effect
\end{subfigure}\hfil 
\begin{subfigure}{0.25\textwidth}
\centering
  Non-Linear Response
\end{subfigure}
\bigskip

\begin{subfigure}{0.25\textwidth}
  \includegraphics[width=\linewidth]{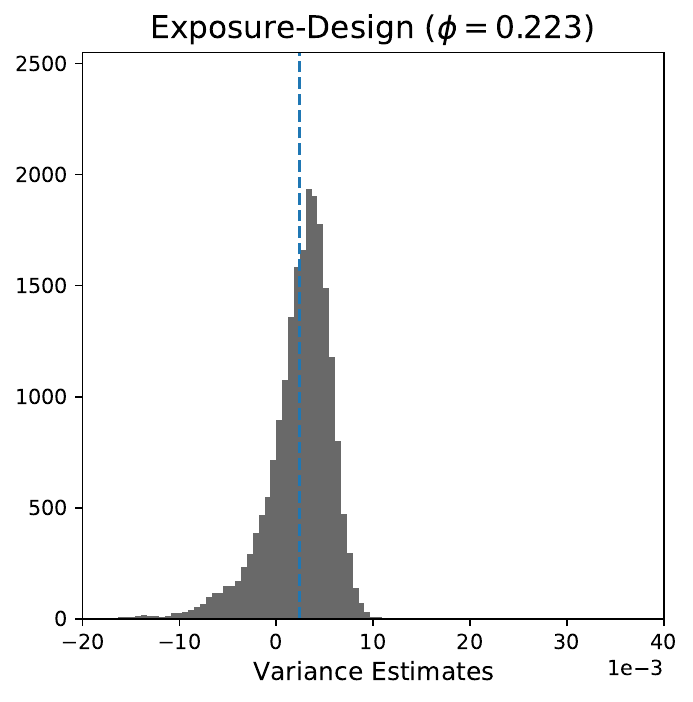}
  \label{fig:var-sim1_d1}
\end{subfigure}\hfil
\begin{subfigure}{0.25\textwidth}
  \includegraphics[width=\linewidth]{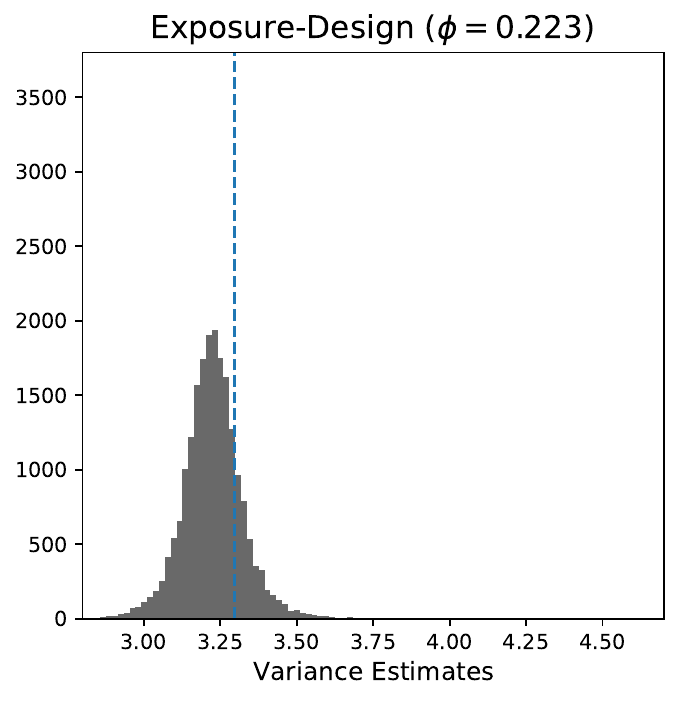}
  \label{fig:var-sim2_d1}
\end{subfigure}\hfil 
\begin{subfigure}{0.25\textwidth}
  \includegraphics[width=\linewidth]{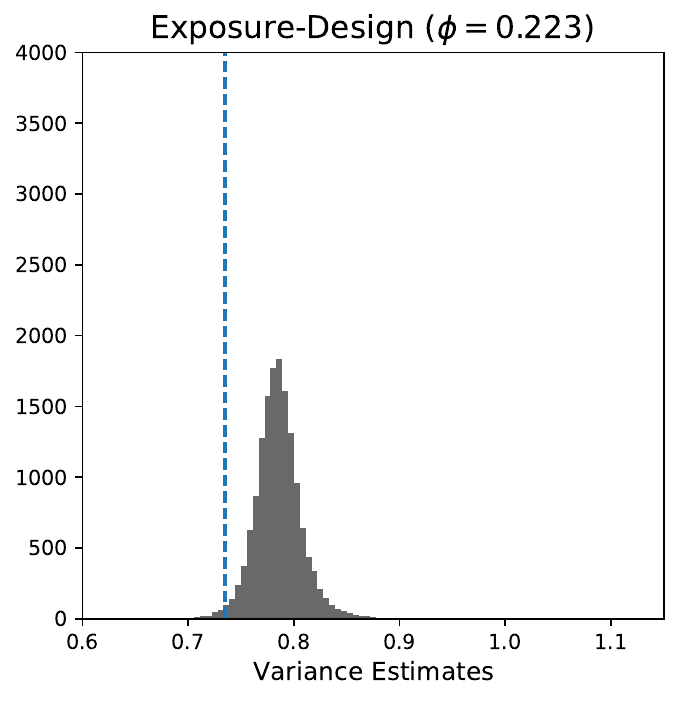}
  \label{fig:var-sim3_d1}
\end{subfigure}

\vspace{-.3cm}
\begin{subfigure}{0.25\textwidth}
  \includegraphics[width=\linewidth]{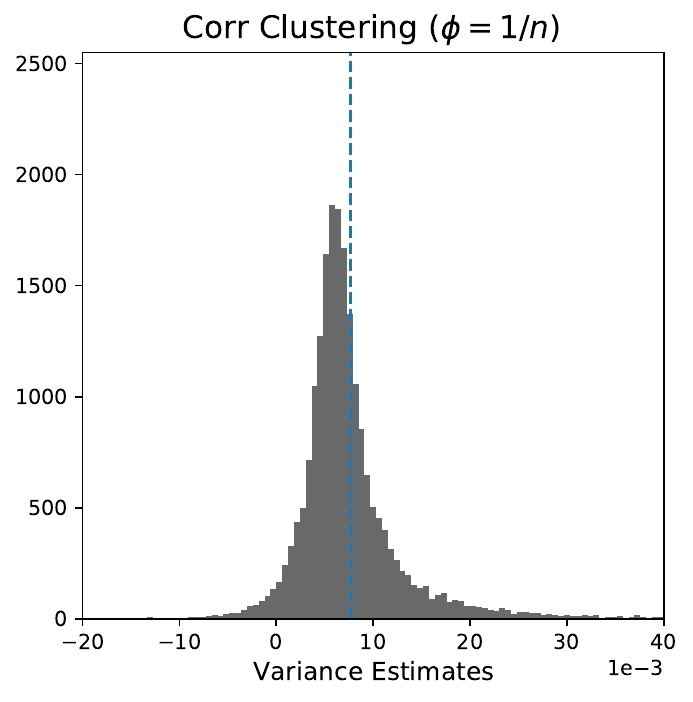}
  \label{fig:var-sim1_d2}
\end{subfigure}\hfil 
\begin{subfigure}{0.25\textwidth}
  \includegraphics[width=\linewidth]{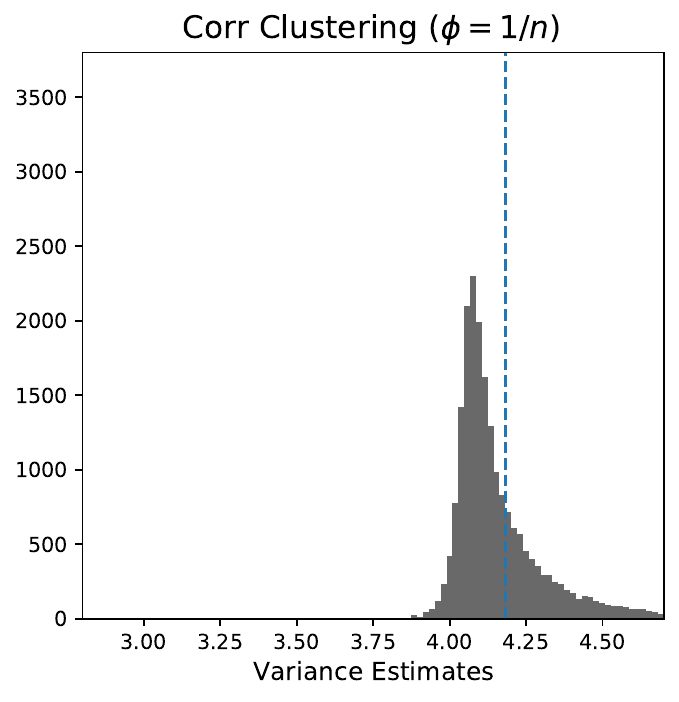}
  \label{fig:var-sim2_d2}
\end{subfigure}\hfil
\begin{subfigure}{0.25\textwidth}
  \includegraphics[width=\linewidth]{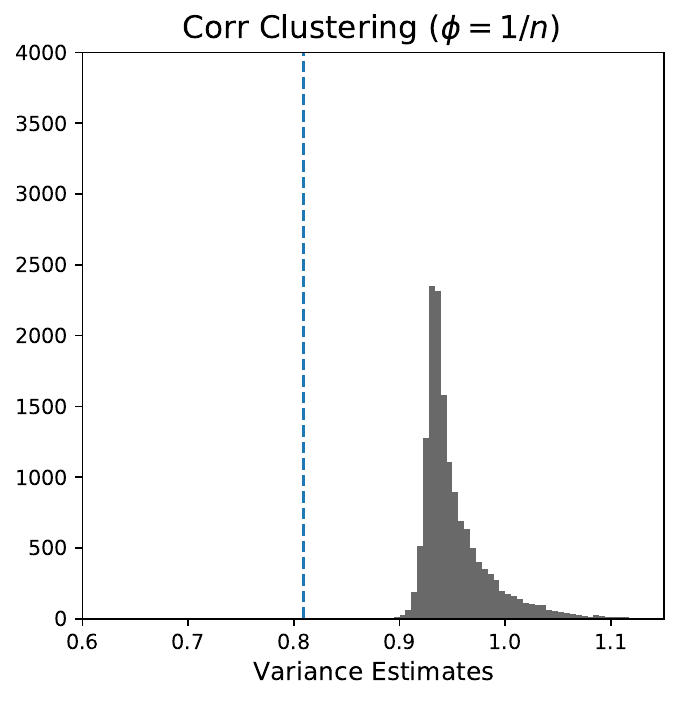}
  \label{fig:var-sim3_d2}
\end{subfigure}

\vspace{-.3cm}
\begin{subfigure}{0.25\textwidth}
  \includegraphics[width=\linewidth]{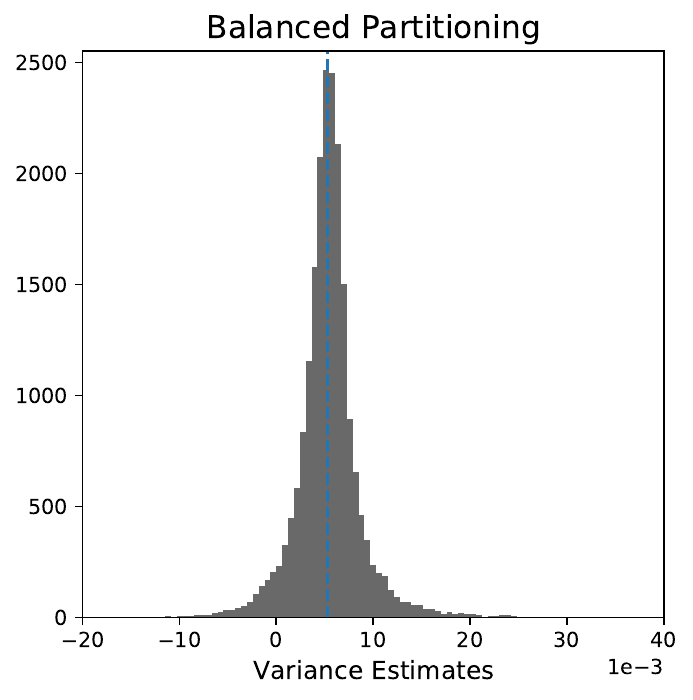}
  \label{fig:var-sim1_d3}
\end{subfigure}\hfil 
\begin{subfigure}{0.25\textwidth}
  \includegraphics[width=\linewidth]{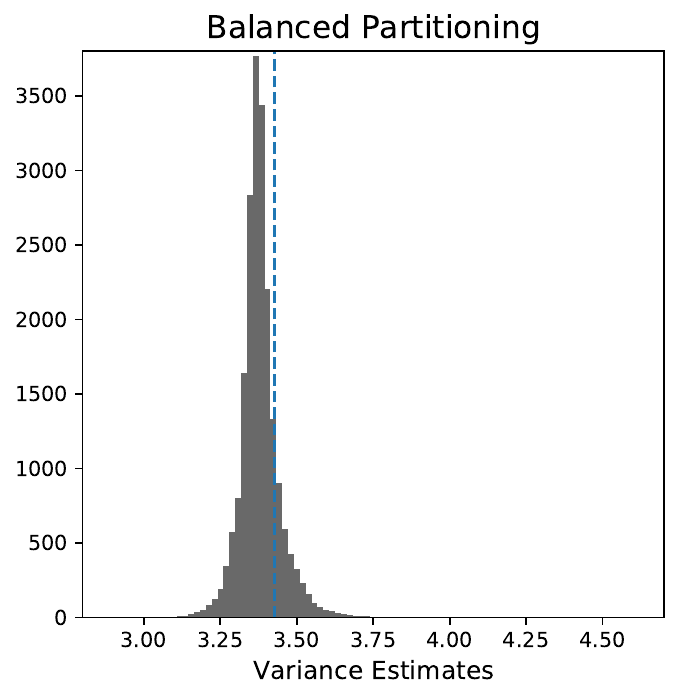}
  \label{fig:var-sim2_d3}
\end{subfigure}\hfil
\begin{subfigure}{0.25\textwidth}
  \includegraphics[width=\linewidth]{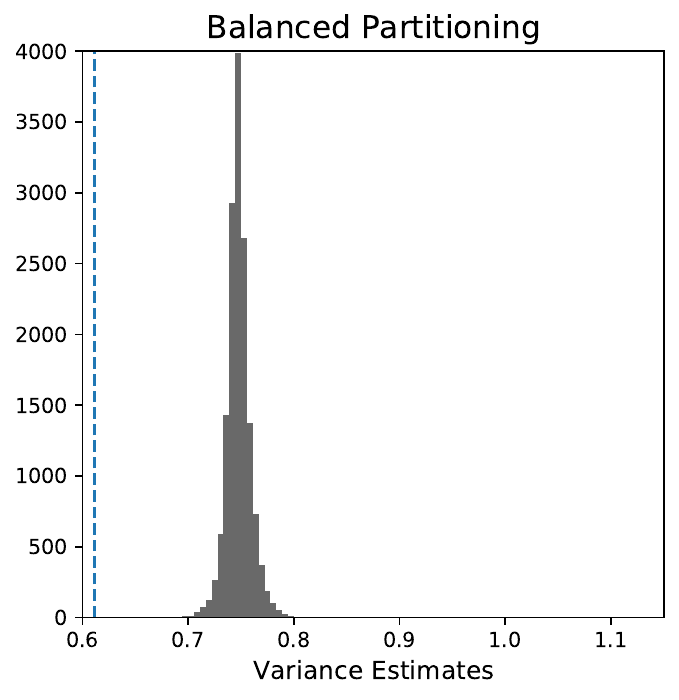}
  \label{fig:var-sim3_d3}
\end{subfigure}
\caption{Histograms of the variance estimator in simulations.}
\label{fig:var-histograms}
\end{figure}

Figure~\ref{fig:var-histograms} contains histograms of the variance estimator for each simulation and treatment design, where the rows correspond to the designs and the columns correspond to the simulations.
The dotted vertical line in the plot is the empirical estimate of the variance of the \doseht estimator, computed from samples.
The variance estimator is unbiased in Simulations~1 and 2, which aligns with Theorem~\ref{thm:variance-estimate}; however, because the empirical variance estimate is used, the blue line is close to (but not exactly) the true variance.
Increasing the number of sampled exposure vectors decreases this error, but drawing more than 20 thousand samples is prohibitively expensive given the size of the data.
In Simulation~1, the mean squared error of all cluster-based designs is so small that the variance estimator takes negative values.
In Simulation 3, the response is highly non-linear so that the variance estimator incurs a positive bias, which results in a coverage slightly above the nominal level.
Interestingly, the variance estimator under the balanced partitioning design is more concentrated around its mean, which is worth further investigation but is beyond the scope of this paper.

\begin{figure}[h!tb]
    \centering 
\begin{subfigure}{0.3\textwidth}
  \includegraphics[width=\linewidth]{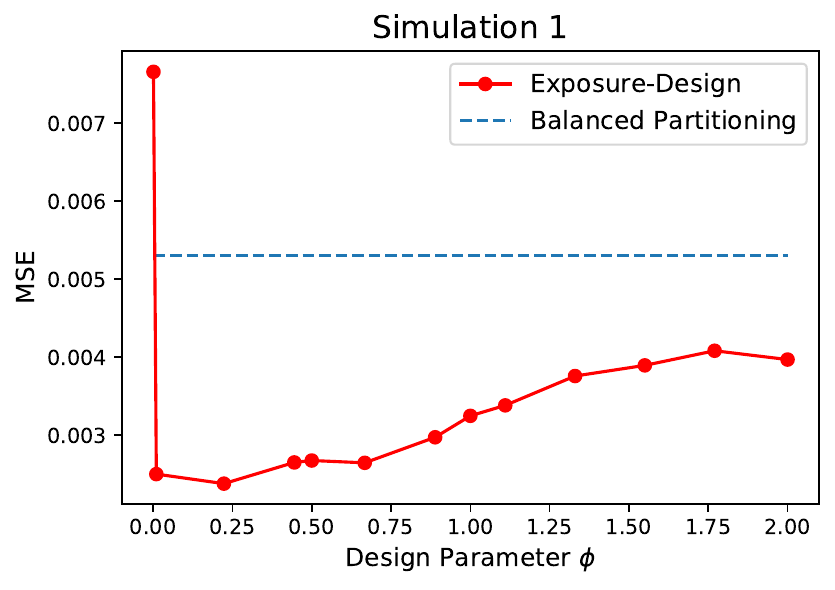}
  \label{fig:mse1}
\end{subfigure}\hfil 
\begin{subfigure}{0.3\textwidth}
  \includegraphics[width=\linewidth]{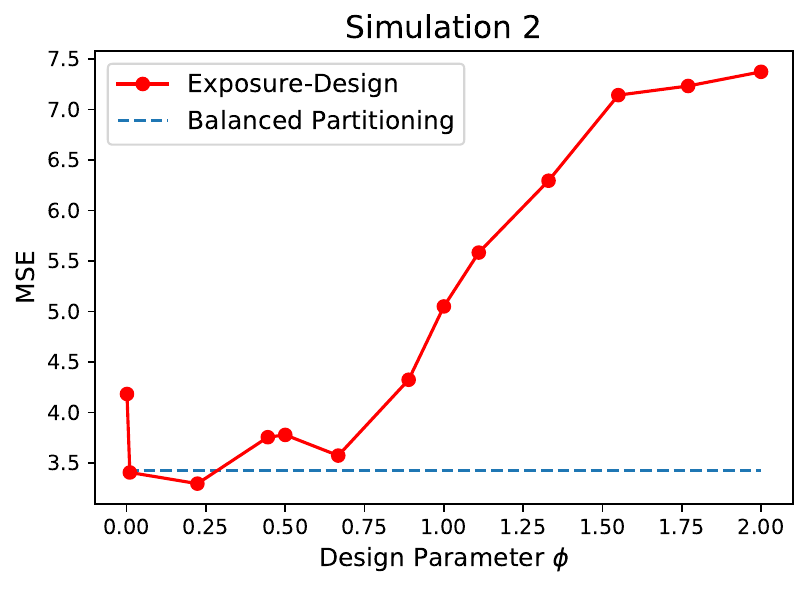}
  \label{fig:mse2}
\end{subfigure}\hfil 
\begin{subfigure}{0.3\textwidth}
  \includegraphics[width=\linewidth]{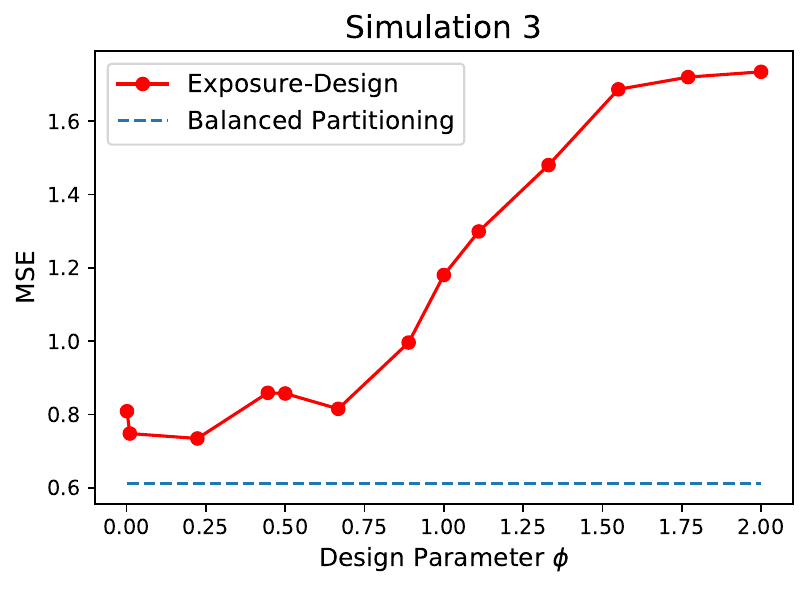}
  \label{fig:mse3}
\end{subfigure}

\caption{MSE of the \doseht estimator as trade-off parameter $\phi$ is varied. First values of $\phi$ are $.001$ and $.01$.}
\label{fig:mse-changing-phi}
\end{figure}

Figure~\ref{fig:mse-changing-phi} contains a plot for each simulation, where the mean squared error is plotted against the correlation penalizing parameter $\phi$.
The mean squared error of the balanced partitioning design appears as a dotted blue line.
In Simulations~1 and 2 where the linear response assumption holds, there is a range of values of $\phi$ where \exposuredesign achieves lower mean squared error than the balanced partitioning design.
In our simulations, the choice of $\phi \approx 1/4$ typically achieves lowest mean squared error.
However, no design is optimal across all types of potential outcomes \citep{Harshaw2021Balancing} and so we encourage experimenters to select $\phi$ (and more generally, select designs) by running tests on simulated data.

\section{Concluding remarks} \label{sec:conclusion}

In this paper, we have presented methodological contributions towards estimating treatment effects in the linear-exposure response model for bipartite experiments: namely, the Exposure Reweighted Linear (\doseht) estimator for consistent and unbiased estimation of average total treatment effect, an unbiased and consistent variance estimator which facilitates construction of asymptotically valid normal-based confidence intervals, and the \exposuredesign that aims to increase precision of this estimator in various settings of interest.

When employing this design in practice, we recommend that the experimenter choose the value of the trade-off parameter $\phi$ by running simulations of the experiment using available models of the outcomes when possible. 
When this is not possible, we find in our simulations that setting $\phi \approx 1/4$ typically yields improvements in the precision of the \doseht estimator over the previously correlation clustering design of \citet{pouget2019variance} where $\phi  = 1/(n-1)$.
We suspect that in most settings of interest, the \doseht estimator will enjoy increased precision under any treatment design which (either explicitly or implicitly) ensures that exposures have large variance and are decorrelated. 

There are several open questions suggested by this work. 
Given that the design we describe in this paper is heuristically motivated, there is likely room for improved designs.
One technical challenge is to construct a design for which consistency and asymptotic normality of the \doseht estimator may be established under weaker growth conditions on the bipartite graph.
Another key open question is developing methods for valid inference in bipartite experiments which go beyond the linear exposure-response assumption.
Finally, it would be of practical and methodological interest to develop estimation techniques that are robust to misspecification in the bipartite graph as well as estimation techniques which perform well in the presence of greater structural interference (i.e. when outcomes are influenced by the exposures of other units).
The results in \citet{Savje2021Causal} regarding estimation of treatment effects under a misspecified exposure mapping might extend to this setting, but that remains to be shown.
Answering these methodological questions around the bipartite experimental framework will increase its relevance and applicability in practice.


\bibliographystyle{apalike}
\bibliography{main.bib}

\newpage
\appendix

\section{Analysis of the \doseht Estimator}\label{sec:estimator-proofs}

In this section, we present proofs of unbiasedness, consistency, and asymptotic normality of \doseht estimator appearing in Section~\ref{sec:estimator} of the main paper.
Before continuing, we introduce some notation used in the proofs.
We begin by defining for each outcome unit $i \in \outunits$, an estimate of the individual treatment effect $\itei$, which is
\[
\dhtesti \triangleq \poi(\zv) \paren[\Big]{ \frac{ \dosei(\zv) - \Exp{\dosei(\zv)}}{\Var{\dosei(\zv)}} } \enspace.
\]
Observe that the \doseht estimator is the average of these estimates of the individual treatment effects, i.e. $\dhtest = (1/n) \sum_{i=1}^n \dhtesti$.
Throughout the proofs, we will often reason about the behavior of the \doseht estimator through the properties of the individual treatment effect estimates.

Next, we introduce the concept of \emph{dependency neighborhoods} \citep{ross2011fundamentals}. 
Let $\erre{1}, \erre{2}, \dots \erre{n}$ be random variables indexed by the integers $[n]$ and collect these random variables into the set $\mathcal{A} = \setb{a_i : i \in [n]}$.
For each variable $a_i$, we define the \emph{dependency neighborhood} as
\[
\depNi \subset \mathcal{A} \text{ such that } a_i \text{ is jointly independent of the variables }  \mathcal{A} \setminus \depNi
\enspace.
\]
In other words, a random variable $a_i$ is jointly independent of all variables not contained in its dependency neighborhood, but is dependent on variables contained in its dependency neighborhood.
We take the convention that $i \in \depNi$ and so that each dependency neighborhood as cardinality at least 1.
A measure of dependence between the random variables is the \emph{maximum dependency degree}, which is $\depdeg = \max_{i \in [n]} \abs{\depNi}$.
Note that independent random variables satisfy $D = 1$ and that completely dependent random variables have $D =n$.

For the remainder of the proof, we focus our discussion of dependency neighborhoods and degrees to the collection of errors of the individual treatment effects,
\[
\erre{1} = \itee{1} - \dhteste{1}, 
\quad \erre{2} = \itee{2} - \dhteste{2}, 
\quad \dots 
\quad \erre{n} = \itee{n} - \dhteste{n} \enspace. 
\]
We begin by showing that in this case, the maximum dependency degree may be bounded in terms of the degrees of the bipartite graph and the dependence in the treatment assignments.

\begin{lemma}\label{lem:dependence-degree-bound}
The dependency degree of the individual treatment effect errors is bounded by $\depdeg \leq \maxcorr \divdeg \outdeg$ .
\end{lemma}
\begin{proof}
The first part of this proof is to establish a necessary condition for an individual treatment effect error $\errj$ to be in the dependency neighborhood of $\erri$, i.e. $\errj \in \depNi$.
We begin by re-writing the exposures under a cluster design.
Recall that the exposures are defined as $\dosei = \sum_{j=1}^m w_{i,j} z_j$.
For each cluster $\cluste{} \in \clustering$, define $w_{i,\cluste{}} = \sum_{j \in \cluste{}} w_{i,j}$ and define $z_{\cluste{}}$ to be the $\pm 1$ cluster treatment assignment variable which is $1$ if diversion units in $\cluste{}$ are treated and $-1$ otherwise.
If $w_{i,\cluste{}} \neq 0$, then we say that cluster $\cluste{}$ is \emph{incident} to outcome unit $i$.
Define $S(i) = \setb{ z_{\cluste{}} :  w_{i,\cluste{}} \neq 0}$ to be the cluster treatment assignments which influence the exposure $\dosei$.
Under the cluster design, the exposure for outcome unit $i$ may be written as
\[
\dosei 
= \sum_{\cluste{} \in \clustering} w_{i,\cluste{}} z_{\cluste{}}
= \sum_{\cluste{} \in S(i)} w_{i,\cluste{}} z_{\cluste{}}
\enspace.
\]
By the linear-response assumption, the individual treatment effect error $\erri$ is a function of the exposure $\dosei$.
Moreover, $\erri$ is a function of the cluster treatment assignment variables in $S(i)$.
Let us denote this relationship by writing $\erri = g_i( S(i) )$, where $g_i$ is a function of the cluster treatment variables $ z_{\cluste{}} \in S(i)$.
Let $B \subset \outunits$ be a collection of outcome units.
We remark that joint independence of cluster treatment assignments implies joint independence of individual treatment effect errors:
\[
S(i) \indep \setb{S(j) : j \in B}
\Rightarrow
\erri \indep \setb{\errj : j \in B} \enspace.
\]
Under an independent cluster design, the cluster treatment assignments $S(i)$ are jointly independent of the cluster treatment assignments $\setb{S(j) : j \in B}$ when the corresponding sets of clusters are disjoint, i.e. $S(i) \cap \paren{ \cup_{j \in B} S(j)} = \emptyset$.
Thus, the individual treatment effect estimate $a_i$ is jointly independent of the collection of individual treatment effect estimates $\setb{a_j : j \in B}$ when outcome unit $i$ is not incident to any cluster that is incident to an outcome unit in $B$.
In other words, $\errj \in \depNi$ only if outcome units $i$ and $j$ are incident to a common cluster.

Fix an outcome unit $i \in \outunits$.
The remainder of the proof is a simple counting argument which uses this necessary condition to establish that $\abs{\depNi} \leq \maxcorr \divdeg \outdeg$.
In particular, we will count the number of outcome units that are incident to one of the clusters that are incident to $i$.
Because the degree of outcome unit $i$ is at most $\outdeg$, it is incident to at most $\outdeg$ clusters.
Each of these clusters has at most $\maxcorr$ diversion units, by Assumption~\ref{assumption:design}.
Because the degree of all diversion units $j$ is at most $\divdeg$, the number of outcome units which are incident to at least one of these clusters is at most $\maxcorr \divdeg \outdeg$.
Thus, we have established that
\[
D 
= \max_{i \in \outunits} \abs{\depNi} 
\leq \maxcorr \divdeg \outdeg \enspace. \qedhere
\]
\end{proof}

The following lemma derives a lower bound the exposure variances in terms of the treatment assignment probability and the maximum degree of the outcome units.

\begin{lemma}\label{lem:dose-variance-lower-bound}
If each pair of treatment assignments is non-negatively correlated, then each exposure variance is lower bounded as $\Var{\dosei} \geq \frac{p (1-p)}{\outdeg}$.
\end{lemma}
\begin{proof}
We begin by expanding the variance of the exposure $\dosei$ by
\begin{align*}
\Var{\dosei}
&= \Var{\sum_{j=1}^m w_{i,j} \zj}
    &\text{(definition of exposure)} \\
&= \sum_{i=1}^m \bracket[\Big]{ \Var{w_{i,j} \zj} + \sum_{\ell \neq j} \Cov{ w_{i,j} \zj, w_{i,\ell} \zl} }
    &\text{(properties of variance)} \\
&= \sum_{i=1}^m \bracket[\Big]{ w_{i,j}^2 \Var{\zj} + \sum_{\ell \neq j} w_{i,j} w_{i,\ell} \Cov{\zj,\zl} } 
    &\text{(properties of variance and covariance)} \\
&\geq \sum_{i=1}^m w_{i,j}^2 \Var{\zj} 
    &\text{(non-negative weights and correlations)}\\
&= p (1-p) \sum_{i=1}^m w_{i,j}^2 \enspace,
\end{align*}
where the inequality follows because the weights $w_{i,j}$ are non-negative and the assignments are non-negatively correlated and the last equality follows because $\zi$ are $0,1$ random variables with $\Pr(\zi = 1) = p$.

We complete the proof by lower bounding the sum of the squares of the weights.
Recall that the sum of the weights is $1$ and there are at most $\outdeg$ non-negative terms in the sum.
Using this together with the inequality that relates $\ell_2$ to $\ell_1$ norms in $d$-dimensions, $\norm{\cdot}_2^2 \geq \frac{1}{d} \norm{\cdot}_1^2$, we have that
\[
\sum_{i=1}^m w_{i,j}^2 
\geq \frac{1}{\outdeg} \sum_{i=1}^m w_{i,j}
= \frac{1}{\outdeg} \enspace. \qedhere
\]
\end{proof}

The following lemma is a bound on the moments of the errors of the individual treatment effects.

\begin{lemma}\label{lem:ite-error-moments}
The $s$th moment of the error of the individual treatment effect estimates is bounded by
\[
\Exp{ \abs{\itei - \dhtesti}^s }
\leq \bracket[\Bigg]{\maxpo \paren[\Big]{2 + \frac{\outdeg}{p(1-p)}}}^s \enspace.
\]
\end{lemma}
\begin{proof}
We begin by remarking that $\abs{\poi(\zv)} \leq \maxpo$ implies that each of the individual slopes are also bounded in absolute value as $\abs{\slopei} \leq 2 \maxpo$.
Recall that by the linear response assumption, $\poi(\zv) = \slopei \dosei + \incepti$ and by the linear exposure assumption (along with the normalization of the edge weights), setting $\zv = \zerovec, \onevec$ results in an exposure of $\xi = 0,1$.
Thus, when considering $\zv = \zerovec, \onevec$, the bound $\abs{\poi(\zv)} \leq M$ implies that $\abs{\slopei + \incepti} \leq \maxpo$ and $\abs{\incepti } \leq \maxpo$, which is enough to establish that $\abs{\slopei} \leq 2\maxpo$.

We now proceed by proving a bound on $\abs{\itei - \dhtesti}$, which holds for any realization of the random variables:
\begin{align*}
\abs{\itei - \dhtesti}
&= \abs[\Big]{\slopei - \poi(\zv) \paren[\Big]{\frac{\dosei - \Exp{\dosei}}{\Var{\dosei}}} } \\
&\leq \abs{\slopei} + \frac{\abs{\poi(\zv)} \cdot \abs{\dosei - \Exp{\dosei}}}{\Var{\dosei}} 
    &\text{(triangle inequality)}\\
&\leq 2 \maxpo + \frac{\maxpo}{\Var{\dosei}} 
    &\text{(definition of $\maxpo$ and above)}\\
&\leq \maxpo \paren[\Big]{ 2 + \frac{1}{\Var{\dosei}} } 
    &\text{(collecting terms)}\\
&\leq \maxpo \paren[\Big]{2 + \frac{\outdeg}{p(1-p)}}
    &\text{(Lemma~\ref{lem:dose-variance-lower-bound})}
\end{align*}
The moment bound follows by applying the bound above.
\end{proof}

\subsection{Expectation of the \doseht estimator (Theorems~\ref{thm:exposure-ht-unbiased} and \ref{thm:expectation-no-assumption}) }

In this section, we derive the expectation of the \doseht estimator, both with and without the linear exposure-response assumption.
First, we derive the expectation under the linear exposure-response assumption.

\begin{manualtheorem}{\ref{thm:exposure-ht-unbiased}}
	\exposurehtunbiased
\end{manualtheorem}
\begin{proof}
By linearity, the expectation of the estimator is 
\[
\Exp{\dhtest}
= \frac{1}{n} \sum_{i=1}^n \Exp[\Bigg]{  \poi(\zv) \paren[\Bigg]{ \frac{\dosei(\zv) - \Exp{\dosei(\zv)}}{{ \Var{\dosei(\zv)}} } } }
\enspace.
\]
By Proposition~\ref{prop:ate-under-linear-exposure-response}, the ATTE is the average of the slope terms $\slopei$.
Thus, to complete the proof we show that each expectation terms inside the sum is equal to the corresponding slope $\slopei$.
Using the linear response assumption,
\begin{align*}
\Exp[\Bigg]{  \poi(\zv) \paren[\Bigg]{ \frac{\dosei(\zv) - \Exp{\dosei(\zv)}}{{ \Var{\dosei(\zv)}} } } }
&= \Exp[\Bigg]{  \paren{\slopei \dosei(\zv) + \incepti} \paren[\Bigg]{ \frac{\dosei(\zv) - \Exp{\dosei(\zv)}}{{ \Var{\dosei(\zv)}} } } } \\
&= \slopei \Exp[\Bigg]{ \dosei(\zv) \paren[\Bigg]{ \frac{\dosei(\zv) - \Exp{\dosei(\zv)}}{{ \Var{\dosei(\zv)}} } } } 
    + \incepti \Exp[\Bigg]{ \paren[\Bigg]{ \frac{\dosei(\zv) - \Exp{\dosei(\zv)}}{{ \Var{\dosei(\zv)}} } } } \\
&= \slopei \paren[\Bigg]{ \frac{\Exp{\dosei(\zv)^2} - \Exp{\dosei(\zv)}^2 }{\Var{\dosei(\zv)}}}
    + \incepti \paren[\Bigg]{ \frac{\Exp{\dosei(\zv)} - \Exp{\dosei(\zv)} }{\Var{\dosei(\zv)}}} \\
&= \slopei 
\qedhere
\end{align*}
\end{proof}

Next, we derive the expectation of the \doseht estimator under a general (non-linear) response assumption.

\begin{manualtheorem}{\ref{thm:expectation-no-assumption}}
	\expectationnoassumption
\end{manualtheorem}
\begin{proof}
	We begin by deriving the expectation of an individual term in the \doseht estimator.
	To this end, observe that
	\[
	\Exp{\dhtesti}
	= s\Exp[\Big]{\poi \paren[\Big]{ \frac{\dosei - \Exp{\dosei}}{\Var{\dosei}} } }
	= \frac{\Exp{\poi \dosei} - \Exp{\poi} \Exp{\dosei}}{\Var{\dosei}}
	= \frac{\Cov{\dosei, \poi}}{\Var{\dosei}} \enspace.
	\]
	The proof is completed by linearity of expectation.
\end{proof}

We remark that Theorem~\ref{thm:exposure-ht-unbiased} follows from Theorem~\ref{thm:expectation-no-assumption} by observing that under a linear response assumption that $\poi = \slopei \dosei + \incepti$, we have that $\Cov{\dosei, \poi} = \Cov{\dosei, \slopei \dosei + \incepti} = \slopei \Var{\dosei}$.

\subsection{Consistency of \doseht estimator (Theorem~\ref{thm:consistency})} \label{sec:consistency-proof}

We are now ready to establish the consistency of the \doseht estimator.
Before doing so, we restate the theorem here.

\begin{manualtheorem}{\ref{thm:consistency}}
\consistencythm
\end{manualtheorem}
\begin{proof}
We begin by proving a finite sample bound on the mean squared error of the \doseht estimator, and then we finish the proof by taking the limit in the asymptotic sequence.
Note that the mean squared error may be broken down into the errors of the individual treatment effect estimates via
\[
\Exp{(\ate - \dhtest)^2}
= \Exp[\Big]{\paren[\Big]{ \frac{1}{n} \sum_{i=1}^n \paren{ \itei - \dhtesti} }^2 }
= \frac{1}{n^2} \sum_{i=1}^n \sum_{j=1}^n \Exp{(\itei - \dhtesti)(\itej - \dhtestj)}
\enspace.
\]
Note that the term in the inner sum is the covariance of the errors in the individual treatment effect estimators.
By definition of the dependency neighborhoods, only terms $j \in \depNi$ are dependent and so only these terms will have non-zero covariance.
Using this and the second moment bound in Lemma~\ref{lem:ite-error-moments}, we have that
\begin{align*}
\Exp{(\ate - \dhtest)^2}
&= \frac{1}{n^2} \sum_{i=1}^n \sum_{j \in \depNi} \Exp{(\itei - \dhtesti)(\itej - \dhtestj)} 
    &\text{(dependency neighborhoods)} \\
&\leq \frac{1}{n^2} \sum_{i=1}^n \sum_{j \in \depNi} \sqrt{ \Exp{(\itei - \dhtesti)^2} \Exp{(\itej - \dhtestj)^2}}
    &\text{(Cauchy-Schwarz)} \\
&\leq \frac{1}{n^2} \sum_{i=1}^n \abs{\depNi} \cdot \bracket[\Big]{M \paren[\Big]{2 + \frac{\outdeg}{p(1-p)}} }^2
    &\text{(Lemma~\ref{lem:ite-error-moments})} \\
&\leq \frac{\depdeg}{n} \bracket[\Big]{M \paren[\Big]{2 + \frac{\outdeg}{p(1-p)}} }^2 \enspace.
    &\text{(max dependency degree)}
\end{align*}
By using the bound $\depdeg \leq \maxcorr \divdeg \outdeg$ given in Lemma~\ref{lem:dependence-degree-bound}, we have the finite-sample bound on the mean squared error:
\[
\Exp{(\ate - \dhtest)^2}
\leq \frac{\maxcorr \divdeg \outdeg  }{n} \bracket[\Big]{M \paren[\Big]{2 + \frac{\outdeg}{p(1-p)}} }^2
= \bigO[\big]{\divdeg \outdeg^3 / n} \enspace,
\]
where the final equality comes from interpreting the finite-sample bound in the context of the asymptotic sequence.
In particular, by Assumptions~\ref{assumption:bounded-potential-outcomes} and \ref{assumption:design}, we have that $M$ is a constant, $p$ is bounded away from $0$ and $1$ by a constant, and $\maxcorr$ is a constant.
\end{proof}

\subsection{Asymptotic normality of the \doseht estimator (Theorem~\ref{thm:asymptotic-normality})} \label{sec:normality-proof}

We establish asymptotic normality of the \doseht estimator by using Stein's method.
In particular, we use the following result from \citet{ross2011fundamentals}:

\begin{lemma}[Lemma 3.6 of \citet{ross2011fundamentals}] \label{lem:steins-method}
Let $\erre{1}, \erre{2}, \dots \erre{n}$ be random variables such that $\Exp{\erri^4} < \infty$, $\Exp{\erri} = 0$, $\sigma^2 = \Var{\frac{1}{n} \sum_{i=1}^n \erri}$, and define $X = \paren{\frac{1}{n} \sum_{i=1}^n \erri} / \sigma$.
Then for a standard normal $Z \sim \mathcal{N}(0,1)$, we have
\[
d_W(X,Z) \leq 
\frac{D^2}{\sigma^3 n^3} \sum_{i=1}^n \Exp{\abs{\erri}^3} 
+ \sqrt{\frac{28}{\pi}} \cdot \frac{\depdeg^{3/2}}{n^2 \sigma^2} \sqrt{ \sum_{i=1}^n \Exp{\erri^4} }
\enspace,
\]
where $\depdeg$ is the maximum dependency degree of the random variables and $d_W(\cdot, \cdot)$ is the Wasserstein distance.
\end{lemma}

We will use Lemma~\ref{lem:steins-method} to prove asymptotic normality of the \doseht estimator.
Before continuing, let us restate the theorem.

\begin{manualtheorem}{\ref{thm:asymptotic-normality}}
\asymptoticanormalitythm
\end{manualtheorem}
\begin{proof}
Our strategy may be described in two main steps: first, we use Lemma~\ref{lem:steins-method} to derive a finite-sample bound on the Wasserstein distance between the distribution of $(\ate - \dhtest) / \sqrt{\Var{\dhtest}}$ and a standard normal. Next, we us this bound to argue that this Wasserstein distance approaches $0$ in the limit of the asymptotic sequence under the above conditions.

We seek to apply Lemma~\ref{lem:steins-method} where the random variables are the errors of the individual treatment effect estimates; that is,
\[
\erre{1} = \itee{1} - \dhteste{1}, \erre{2} = \itee{2} - \dhteste{2}, \dots \erre{n} = \itee{n} - \dhteste{n} \enspace. 
\]
Note that $\frac{1}{n} \sum_{i=1}^n \erri = \frac{1}{n} \sum_{i=1}^n \itei - \dhtesti = \ate - \dhtest$ and $\Var{\frac{1}{n} \sum_{i=1}^n \erri} = \Var{\dhtest}$ so that the random variable $X$ in Lemma~\ref{lem:steins-method} is equal to $(\ate - \dhtest) / \sqrt{\Var{\dhtest}}$, which is indeed the random variable we wish to characterize.
Let us show that the conditions of Lemma~\ref{lem:steins-method} are satisfied: first, recall that $\dhtesti$ are unbiased estimates of $\itei$ so that $\erri$ has mean zero.
Second, because the potential outcomes are bounded by a constant $\maxpo$, the support of $\erri$ is bounded so the fourth moments are finite.
Thus, we may apply Lemma~\ref{lem:steins-method} in this setting.

We will use the Lemma~\ref{lem:ite-error-moments} to bound the sum of the third and fourth moments.
In particular, Lemma~\ref{lem:ite-error-moments} implies that
\[
\sum_{i=1}^n \Exp{\abs{\erri}^3} \leq n \cdot \bracket[\Bigg]{ \maxpo \paren[\Big]{2 + \frac{\outdeg}{p(1-p)}}}^3 
\quad \text{ and } \quad
\sum_{i=1}^n \Exp{\abs{\erri}^4} \leq n \cdot \bracket[\Bigg]{ \maxpo \paren[\Big]{2 + \frac{\outdeg}{p(1-p)}}}^4
\enspace.
\]
Using this moment bound together with the bound on the maximum dependence degree $\depdeg$ (Lemma~\ref{lem:dependence-degree-bound}) on the result of Lemma~\ref{lem:steins-method}, we obtain that
\[
d_W\paren[\Big]{\frac{\ate - \dhtest}{\sqrt{\Var{\dhtest}}}, Z} \leq
\frac{(\maxcorr \divdeg \outdeg)^2}{\sigma^3 n^2} \cdot \bracket[\Bigg]{\maxpo \paren[\Big]{2 + \frac{\outdeg}{p(1-p)}}}^3 
+
\sqrt{\frac{28}{\pi}} \cdot
\frac{(\maxcorr \divdeg \outdeg)^{3/2}}{\sigma^2 n^{3/2}} \cdot \bracket[\Bigg]{\maxpo \paren[\Big]{2 + \frac{\outdeg}{p(1-p)}}}^2
\]

We now interpret this finite-sample bound in the context of the asymptotic sequence.
By Assumptions~\ref{assumption:bounded-potential-outcomes} and \ref{assumption:design}, we have that $M$ is a constant, $p$ is bounded away from $0$ and $1$ by a constant, and $\maxcorr$ is a constant.
It follows that the Wasserstein distance between $(\ate - \dhtest) / \sqrt{\Var{\dhtest}}$ and a standard normal is asymptotically bounded as
\[
d_W\paren[\Big]{\frac{\ate - \dhtest}{\sqrt{\Var{\dhtest}}}, Z}
= \bigO[\Big]{
\frac{ \divdeg^2 \outdeg^5}{\sigma^3 n^2}
+ \frac{\divdeg^{3/2} \outdeg^{7/2}}{\sigma^2 n^{3/2}}
}
\]
By Assumption~\ref{assumption:mse-rate}, we have that $\Var{\dhtest} = \bigOmega{1 / n}$, which means that this bound becomes
\[
d_W\paren[\Big]{\frac{\ate - \dhtest}{\sqrt{\Var{\dhtest}}}, Z}
= \bigO[\Big]{
\frac{ \divdeg^2 \outdeg^5}{n^{1/2}}
+ \frac{\divdeg^{3/2} \outdeg^{7/2}}{n^{1/2}}
}
= \bigO[\Big]{\frac{ \divdeg^2 \outdeg^5}{n^{1/2}}}
\enspace.
\]
By assumption, the asymptotic sequence satisfies $\divdeg^{4} \outdeg^{10} = \littleO{n}$.
Thus, the Wasserstein distance between $(\ate - \dhtest) / \sqrt{\Var{\dhtest}}$ and a standard normal approaches $0$ in this asymptotic sequence.

\end{proof}
\section{Variance Estimation and Confidence Intervals}\label{sec:interval-proofs}

In this section, we present the proofs of unbiasedness and consistency of the variance estimator together with a proof of the asymptotic validity of the normal based confidence intervals.

\subsection{Closed form expressions for coefficients}\label{sec:closed-form-coefficients}

We begin by deriving closed form expressions for the coefficients $a_{i,j}$, $b_{i,j}$, and $c_{i,j}$ which are obtained as solutions to the system of linear equations:

\varestlinearsystem

\paragraph{Distinct Outcome Units $(i \neq j)$} Note that the matrix in the linear system above is the 3-by-3 covariance matrix of the exposures $\dosei$, $\dosej$ and their product $\dosei \dosej$. Recall that we have defined this matrix to be $\covmatij$. A unique solution exists if and only if $\det(\covmatij) > 0$, which is the say that the generalized variance of $\dosei$, $\dosej$, and $\dosei \dosej$ is nonzero.
In this case, the solution coefficients may be explicitly derived via Cramer's rule for solving linear systems as
\begin{align*}
a_{i,j} & = \frac{\Var{\dosei} \Var{\dosej} - \Cov{\dosei,\dosej}^2}{\det(\covmatij)}, \\
b_{i,j} & = \frac{\Cov{\dosei,\dosej} \Cov{ \dosei \dosej, \dosej} - \Var{\dosej} \Cov{\dosei \dosej, \dosei}}{\det(\covmatij)},\\
c_{i,j} & = \frac{\Cov{\dosei,\dosej} \Cov{ \dosei \dosej, \dosei} - \Var{\dosei} \Cov{\dosei \dosej, \dosej}}{\det(\covmatij)}
\enspace.
\end{align*}
The determinant $\det(\covmatij)$ is a polynomial in the entries of the matrix $\covmatij$.
For completeness, we present the determinant calculation:
\begin{align*}
    \det(\covmatij) 
    &= \Var{\dosei \dosej} \paren{ \Var{\dosei}\Var{\dosej} - \Cov{\dosei, \dosej}^2} - \Var{\dosei} \Cov{\dosei \dosej, \dosej}^2  \\
    & \quad \quad
    - \Var{\dosej} \Cov{\dosei \dosej, \dosei}^2 
    + 2 \Cov{\dosei, \dosej} \Cov{\dosei \dosej,\dosej} \Cov{\dosei \dosej, \dosei}  \enspace.
\end{align*}

\paragraph{Same Outcome Unit ($i=j$)} 
Note that the $3$-by-$3$ covariance matrix defining the system of linear equation above will have zero determinant, as $\dosei = \dosej$.
Nevertheless, we show that a solution to the linear system may still be obtained under certain conditions on the distribution of the exposure $\dosei$.
Observe that the $b_{i,i}$ and $c_{i,i}$ terms are redundant, as $\dosei = \dosej$.
By taking $c_{i,i} = 0$, we reduce the 3-by-3 linear system in Proposition~\ref{prop:unbiased-individual-covest} to the following 2-by-2 linear system:
\[
\begin{bmatrix}
\Var{\dosei^2} & \Cov{\dosei, \dosei^2} \\
\Cov{\dosei^2, \dosei} & \Var{\dosei}
\end{bmatrix}
\begin{bmatrix}
a_{i,i} \\
b_{i,i}
\end{bmatrix}
=
\begin{bmatrix}
1 \\
0
\end{bmatrix}
\]
Observe that the matrix in this linear system is the covariance matrix between the exposure $\dosei$ and its square $\dosei^2$.
Recall that in the main body, we defined this 2-by-2 matrix to be $\covmatii$.
Note that a solution exists when the $\det(\covmatii) > 0$, which is to say that the exposure $\dosei$ and its square $\dosei^2$ are not perfectly correlated.
In this case, the solution to the coefficients may be obtained as
\[
a_{i,i} = \frac{\Var{\dosei}}{\det(\covmatii)}, \quad
b_{i,i} = \frac{\Cov{\dosei, \dosei^2}}{\det(\covmatii)}, \quad 
c_{i,i} = 0
\quad
\text{where} \quad
\det(\covmatii) = \Var{\dosei} \Var{\dosei^2} - \Cov{\dosei, \dosei^2}^2 
\enspace.
\]

\subsection{Unbiasedness of variance estimator (Theorem~\ref{thm:variance-estimate})}

In this section, we present proofs of Proposition~\ref{prop:unbiased-individual-covest} and Theorem~\ref{thm:variance-estimate}, which establish unbiasedness of the proposed variance estimator under certain conditions on the exposure distribution.

\newcommand{\unbiasedcovestprop}{
Under the linear response assumption, $\covestij = \poi(\zv) \poj(\zv) R_{i,j}(\dosei, \dosej)$ is an unbiased estimator for the individual covariance term $\Cov{\dhtesti, \dhtestj}$ if the coefficients $a_{i,j}, b_{i,j}, c_{i,j}$ in the weighting function $S_{i,j}(\dosei,\dosej)$ satisfy the system of linear equations:
\varestlinearsystem
}
\begin{prop}\label{prop:unbiased-individual-covest}
\unbiasedcovestprop
\end{prop}

\begin{proof}

If $\Cov{\dosei, \dosej} = 0$, then by the linear response assumption $\Cov{\dhtesti, \dhtestj} = 0$.
In this case, the weighting function is identically zero and thus $\covestij = 0$, which is trivially unbiased.
The remainder of the proof focuses on the case where $\Cov{\dosei, \dosej} \neq 0$ and the weighting function $R_{i,j}(\dosei, \dosej) = Q_{i,j}(\dosei, \dosej) - S_{i,j}(\dosei, \dosej)$.
Observe that the expectation of the individual covariance estimator $\covestij$ is equal to

\begin{equation}\label{eq:expanding-exp-covestij}
\Exp[\Big]{\covestij}
= \Exp[\Big]{\poi(\zv) \poj(\zv) R_{i,j}(\dosei, \dosej)}
= \Exp[\Big]{\poi(\zv) \poj(\zv) Q_{i,j}(\dosei, \dosej)}
- \Exp[\Big]{\poi(\zv) \poj(\zv) S_{i,j}(\dosei, \dosej)}
\enspace.
\end{equation}

By construction of the weighting function $Q_{i,j}(\dosei, \dosej)$, we can compute the expectation of the first term in \eqref{eq:expanding-exp-covestij} as 
\[
\Exp[\Big]{\poi(\zv) \poj(\zv) Q_{i,j}(\dosei, \dosej)}
= \Exp[\Bigg]{\poi(\zv) \poj(\zv) \paren[\Big]{\frac{\dosei - \Exp{\dosei}}{\Var{\dosei}}} \paren[\Big]{\frac{\dosej - \Exp{\dosej}}{\Var{\dosej}} } }
= \Exp[\big]{\dhtesti \dhtestj} \enspace.
\]

Next, we evaluate the expectation of the second term in \eqref{eq:expanding-exp-covestij}.
Before doing so, observe that $\Exp{S_{i,j}(\dosei, \dosej)} = 0$ by construction.
Moreover, the coefficients used in the $S_{i,j}(\dosei, \dosej)$ weighting function satisfy the system of linear equations by assumption, which is equivalent to the following three equations:
\begin{itemize}
    \item $\Exp{\dosei S_{i,j}(\dosei, \dosej)} = 0$
    \item $\Exp{\dosej S_{i,j}(\dosej, \dosej)} = 0$
    \item $\Exp{\dosei \dosej S_{i,j}(\dosei, \dosej)} = 1$
\end{itemize}
Using these four equations together with the linear response assumption, we evaluate the expectation of the second term in \eqref{eq:expanding-exp-covestij} as
\begin{align*}
    &\Exp[\Big]{\poi(\zv) \poj(\zv) S_{i,j}(\dosei, \dosej)} \\
    &\quad = \Exp[\Big]{(\slopei \dosei + \incepti) (\slopej \dosej + \inceptj) S_{i,j}(\dosei, \dosej)} \\
    &\quad= \slopei \slopej \Exp[\Big]{\dosei \dosej S_{i,j}(\dosei, \dosej)}
    + \incepti \inceptj \Exp[\Big]{S_{i,j}(\dosei, \dosej)} 
    + \slopei \inceptj \Exp[\Big]{\dosei S_{i,j}(\dosei, \dosej)}
    + \incepti \slopej \Exp[\Big]{\dosej S_{i,j}(\dosei, \dosej)} \\
    &\quad= \slopei \slopej \\
    &\quad= \Exp{\dhtesti} \Exp{\dhtestj} \enspace,
\end{align*}
where the last inequality follows from the unbiasedness of the individual treatment effect estimators $\dhtesti$ and $\dhtestj$.
Thus, substituting these two calculations into \eqref{eq:expanding-exp-covestij} yields the desired result:
\[
\Exp[\Big]{\covestij}
= \Exp[\big]{\dhtesti \dhtestj}
- \Exp{\dhtesti} \Exp{\dhtestj}
= \Cov{\dhtesti, \dhtestj}
\qedhere
\]
\end{proof}

We are now ready to prove Theorem~\ref{thm:variance-estimate}, which establishes unbiasedness of the variance estimator.
\begin{manualtheorem}{\ref{thm:variance-estimate}}
\varestthm
\end{manualtheorem}

\begin{proof}
Note that Assumption~\ref{assumption:non-degenerate-exposures}, together with the linear response assumption, ensure that the conditions of Proposition~\ref{prop:unbiased-individual-covest} hold for every pair $i,j \in [n]$ so that $\Exp{\covestij} = \Cov{\dhtesti, \dhtestj}$.
Using this fact, we may calculate the expectation of the variance estimate $\erlvarest$ as
\begin{align*}
\Exp{\erlvarest} 
&= \Exp[\Bigg]{
\frac{1}{n^2} \sum_{i=1}^n \sum_{j=1}^n \poi(\zv) \poj(\zv) R_{i,j}(\dosei, \dosej)
} \\
&= \frac{1}{n^2} \sum_{i=1}^n \sum_{j=1}^n \Exp[\Big]{ \poi(\zv) \poj(\zv) R_{i,j}(\dosei, \dosej)} \\
&= \frac{1}{n^2} \sum_{i=1}^n \sum_{j=1}^n \Cov{\dhtesti, \dhtestj} 
= \Var{\dhtest} 
\qedhere
\enspace.
\end{align*}
\end{proof}

\subsection{Consistency of variance estimator (Theorem~\ref{thm:consistent-var-estimation})}

In this section, we present the proof of Theorem~\ref{thm:consistent-var-estimation}, which establishes consistency of the normalized variance estimator.
The main parts of the proof are to show that the individual covariance estimators $\covestij$ are sufficiently uncorrelated and that they have small variance.

To establish a bound on the correlation between covariance estimators $\covestij$, we use the formalism of the dependency graph. 
For each pair of outcome units $i,j \in [n]$, define $\cerrij \triangleq \Cov{\dhtesti, \dhtestj} - \covestij$ to be the error of the individual covariance estimator. 
Let $\mathcal{A}_V = \setb{\cerrij : i,j \in [n]}$ be the set of individual covariance estimator errors.
For each variable $\cerrij$, we define the \emph{dependency neighborhood} as
\[
\depNvar{i,j} \subset \mathcal{A}_V \text{ such that } \cerrij \text{ is jointly independent of the variables }  \mathcal{A}_V \setminus \depNvar{i,j}
\enspace.
\]
Unlike the dependency graph in Appendix~\ref{sec:estimator-proofs}, this dependency graph is indexed over pairs of integers $i,j \in [n]$.
Additionally, quantities associated with this dependency graph are denoted by a subscript $V$.

The following lemma bounds the maximum degree of the dependency graph and is based on a similar counting argument as that used to prove Lemma~\ref{lem:dependence-degree-bound}.

\begin{lemma}\label{lemma:covest-dep-degree-bound}
The dependency degree of individual covariance estimator errors 
is at most $\depdegvar \leq 4 (\maxcorr \divdeg \outdeg)^2$.
\end{lemma}

\begin{proof}
The first part of this proof is to establish a necessary condition for an individual covariance error $\cerrkl$ to be in the dependency neighborhood of $\cerrij$, i.e. $\cerrkl \in \depNvar{i,j}$.
We begin by re-writing the exposures under a cluster design.
Recall that the exposures are defined as $\dosei = \sum_{j=1}^m w_{i,j} z_j$.
For each cluster $\cluste{} \in \clustering$, define $w_{i,\cluste{}} = \sum_{j \in \cluste{}} w_{i,j}$ and define $z_{\cluste{}}$ to be the $\pm 1$ cluster treatment assignment variable which is $1$ if diversion units in $\cluste{}$ are treated and $-1$ otherwise.
If $w_{i,\cluste{}} \neq 0$, then we say that cluster $\cluste{}$ is \emph{incident} to outcome unit $i$.
Define $S(i) = \setb{ z_{\cluste{}} :  w_{i,\cluste{}} \neq 0}$ to be the cluster treatment assignments which influence the exposure $\dosei$.
Under the cluster design, the exposure for outcome unit $i$ may be written as
\[
\dosei 
= \sum_{\cluste{} \in \clustering} w_{i,\cluste{}} z_{\cluste{}}
= \sum_{\cluste{} \in S(i)} w_{i,\cluste{}} z_{\cluste{}}
\enspace.
\]
By the linear exposure-response assumption, the individual covariance estimator error $\cerrij$ is a function of the exposures $\dosei$ and $\dosej$.
Moreover, $\cerrij$ is a function of the cluster treatment assignment variables in $S(i,j) \triangleq S(i) \cup S(i)$.
Let us denote this relationship by writing $\cerrij = g_{i,j}( S(i,j) )$, where $g_{i,j}$ is a function of the cluster treatment variables $ z_{\cluste{}} \in S(i,j)$.
Let $B \subset \outunits \times \outunits$ be a collection of pairs of outcome units.
We remark that joint independence of cluster treatment assignments implies joint independence of individual covariance errors:
\[
S(i,j) \indep \setb{S(k,\ell) : (k,\ell) \in B}
\Rightarrow
\cerrij \indep \setb{\cerrkl : (k,\ell) \in B} \enspace.
\]
Under an independent cluster design, the cluster treatment assignments $S(i,j)$ are jointly independent of the cluster treatment assignments $\setb{S(k,\ell) : (k,\ell) \in B}$ when the corresponding sets of clusters are disjoint, i.e. $S(i,j) \cap \paren{ \cup_{(k,\ell) \in B} S(k,\ell)} = \emptyset$.
Thus, the individual covariance estimate $\cerrij$ is jointly independent of the collection of individual treatment effect estimates $\setb{\cerrkl : (k,\ell) \in B}$ when outcome units $i$ and $j$ are not incident to any cluster that is incident to an outcome unit in $B$.
In other words, $\cerrkl \in \depNvar{j,k}$ only if one of the outcome units $i,j$ and one of the outcome units $k,\ell$ are incident to a common cluster.

Fix a pair of outcome units $i,j \in \outunits$.
The remainder of the proof is a simple counting argument which uses this necessary condition to establish that $\abs{\depNvar{i,j}} \leq (2 \maxcorr \divdeg \outdeg)^2$.
In particular, we will count the number of outcome units that are incident to one of the clusters that are incident to $i$ and $j$.
Because the degrees of outcome unit $i$ and $j$ are at most $\outdeg$, they are incident to at most $2 \outdeg$ clusters.
Each of these clusters has at most $\maxcorr$ diversion units, by Assumption~\ref{assumption:design}.
Because the degree of all diversion units is at most $\divdeg$, the number of pairs of outcome units which are incident to at least one of these clusters is at most
\[
\binom{2 \maxcorr \divdeg \outdeg}{2} \leq ( 2 \maxcorr \divdeg \outdeg)^2 
\enspace.
\]
Thus, we have established that
\[
D 
= \max_{i,j \in \outunits} \abs{\depNvar{i,j}} 
\leq 4 (\maxcorr \divdeg \outdeg)^2 
\enspace. \qedhere
\]
\end{proof}

Next, we derive a bound on the variance of the individual covariance estimators.
Because these individual estimators are unbiased, this yields a bound on their mean squared error.
Recall that $\Delta = \min_{i,j \in [n]} \det(\covmatij)$ is defined to be the smallest non-degeneracy measure amongst all distinct pairs ($i \neq j$) of exposures and single exposures ($i=j$).
We begin by establishing that the coefficients of the weighting function are bounded in magnitude.

\begin{lemma}\label{lemma:bounded-var-coeff}
For each pair of outcome units $i,j \in [n]$, the absolute values of the coefficients in the weighting function $S_{i,j}(\dosei, \dosej)$ are bounded in magnitude as $\max \setb{ \abs{a_{i,j}}, \abs{b_{i,j}}, \abs{c_{i,j}}} \leq 2 / \Delta$.
\end{lemma}
\begin{proof}
Recall that the coefficients $a_{i,j}$, $b_{i,j}$, and $c_{i,j}$ in the weighting function $S_{i,j}(\dosei, \dosej)$ are of the form
\[
a_{i,j} = \frac{\tilde{a}_{i,j}}{\det(\covmatij)},
b_{i,j} = \frac{\tilde{b}_{i,j}}{\det(\covmatij)}, 
\text{ and }
c_{i,j} = \frac{\tilde{c}_{i,j}}{\det(\covmatij)} 
\enspace,
\]
where $\tilde{a}_{i,j}$, $\tilde{b}_{i,j}$, and $\tilde{c}_{i,j}$ depend on statistics of the joint distribution of the exposures $\dosei$ and $\dosej$.
We will now show that $\max \setb{\abs{\tilde{a}_{i,j}}, \abs{\tilde{b}_{i,j}}, \abs{\tilde{c}_{i,j}}} \leq 2$.
We focus only on the case of distinct exposures ($i \neq j$), as the case of a single exposure ($i=j$) follows in an identical way.
First, observe that
\[
\tilde{a}_{i,j} 
= \Var{\dosei} \Var{\dosej} - \Cov{\dosei, \dosej}^2 
\leq \Var{\dosei} \Var{\dosej}
\leq 1 \enspace.
\]
Similarly using the triangle inequality, Cauchy-Schwarz inequality, and the fact that the exposures are supported on $[-1, 1]$, we have that
\begin{align}
    \abs{\tilde{b}_{i,j}}
    &= \abs{\Cov{\dosei, \dosej} \Cov{\dosei \dosej, \dosej} - \Var{\dosej} \Cov{\dosei \dosej, \dosei}} \\
    &\leq \abs{\Cov{\dosei, \dosej} \Cov{\dosei \dosej, \dosej}} + \abs{\Var{\dosej} \Cov{\dosei \dosej, \dosei}} \\
    &\leq \sqrt{\Var{\dosei} \Var{\dosej} \Var{\dosei \dosej} \Var{\dosej}} + \Var{\dosej} \sqrt{\Var{\dosei \dosej} \Var{\dosei} } \\
    &= 2\Var{\dosej} \sqrt{\Var{\dosei} \Var{\dosei \dosej}}\\
    &\leq 2 \enspace,
\end{align}
where the last inequality follows from the fact that random variables on $[0,1]$ have variance at most 1.
The bound on $\abs{\tilde{c}_{i,j}}$ is identical.
This establishes that $\max \setb{\abs{\tilde{a}_{i,j}}, \abs{\tilde{b}_{i,j}}, \abs{\tilde{c}_{i,j}}} \leq 2$.
Thus, we have that $\max \setb{\abs{a_{i,j}}, \abs{b_{i,j}}, \abs{c_{i,j}}} \leq 2 / \det(\covmatij) \leq 2 / \Delta$.
\end{proof}

Using Lemma~\ref{lemma:bounded-var-coeff} together with previously proved lower bounds on the variance of an exposure (Lemma~\ref{lem:dose-variance-lower-bound}), we obtain the following bound on the variance of the individual covariance estimators:

\begin{lemma}\label{lemma:mse-bound-on-covest}
The variance of an individual covariance estimator is bounded by
\[
\Var{\covestij} \leq C \maxpo^4 \paren[\Bigg]{ \paren[\Big]{\frac{\outdeg}{p(1-p)}}^4 + \frac{1}{\Delta^2} }
\]
for some absolute constant $C$.
\end{lemma}

\begin{proof}
Recall that for two random variables $X$ and $Y$, we have the following inequality: $\Var{X+Y} \leq \paren{ \sqrt{\Var{X}} + \sqrt{\Var{Y}} }^2$.
Applying this inequality to the individual covariance estimator, we obtain
\begin{align}
    \Var{\covestij} 
    &= \Var[\Big]{\poi(\zv) \poj(\zv) R_{i,j}(\dosei, \dosej)} \\
    &= \Var[\Big]{\poi(\zv) \poj(\zv) Q_{i,j}(\dosei, \dosej) - \poi(\zv) \poj(\zv) S_{i,j}(\dosei, \dosej)} \\
    &\leq \paren[\Bigg]{
        \sqrt{\Var[\Big]{\poi(\zv) \poj(\zv) Q_{i,j}(\dosei, \dosej)}} 
        + \sqrt{\Var[\Big]{\poi(\zv) \poj(\zv) S_{i,j}(\dosei, \dosej)}} }^2 
\end{align}
Our next goal is to bound each of the terms appearing above.
The variance in the first term may be bounded as
\begin{align}
    \Var{\poi(\zv) \poj(\zv) Q_{i,j}(\dosei, \dosej)}
    &= \Var[\bigg]{\poi(\zv) \poj(\zv) \paren[\Big]{\frac{\dosei - \Exp{\dosei}}{\Var{\dosei}}} \paren[\Big]{\frac{\dosej - \Exp{\dosej}}{\Var{\dosej}}} } \\
    &= \frac{1}{\paren[\big]{\Var{\dosei} \Var{\dosej}}^2} \Var{\poi(\zv) \poj(\zv) (\dosei - \Exp{\dosei}) (\dosej - \Exp{\dosej})} \\
    &\leq \paren[\Big]{\frac{\outdeg}{p (1-p)}}^4 \Var{\poi(\zv) \poj(\zv) (\dosei - \Exp{\dosei}) (\dosej - \Exp{\dosej})} \\
    &\leq \paren[\Big]{\frac{\outdeg}{p (1-p)}}^4 \cdot \paren[\big]{2 \maxpo^2}^2 \\
    \enspace,
\end{align}
where the first inequality follows from using Lemma~\ref{lem:dose-variance-lower-bound} and the second inequality follows from the bound $\abs{\poi(\zv) \poj(\zv) (\dosei - \Exp{\dosei}) (\dosej - \Exp{\dosej})} \leq M^2$. to lower bound the variance of the exposures and the upper bound on the potential outcomes, i.e. $\abs{\poi(\zv)} \leq \maxpo$ for all $\zv \in \setb{0,1}^n$.

We now seek to bound the variance appearing in the second term.
Note that the magnitude of the term inside the variance may be bounded as
\begin{align}
    \abs{\poi(\zv) \poj(\zv) S_{i,j}(\dosei, \dosej)}
    &\leq \abs{\poi(\zv) \poj(\zv)} \cdot \abs{S_{i,j}(\dosei, \dosej)} \\
    &= \abs{\poi(\zv) \poj(\zv)} \cdot
    \abs{
    a_{i,j} \paren{\dosei \dosej - \Exp{\dosei \dosej}}
    + b_{i,j} \paren{\dosei - \Exp{\dosei}}
    + c_{i,j} \paren{\dosej - \Exp{\dosej}} } \\
    &\leq \abs{\poi(\zv) \poj(\zv)} \cdot 
    \paren[\big]{
    \abs{a_{i,j}} \cdot \abs{\dosei \dosej - \Exp{\dosei \dosej}}
    + \abs{b_{i,j}} \cdot \abs{\dosei - \Exp{\dosei}}
    + \abs{c_{i,j}} \cdot \abs{\dosej - \Exp{\dosej}}
    } \\
    &\leq M^2 \cdot 3 \cdot (2 / \Delta) \cdot 2 \\
    &= 12 \frac{M^2}{\Delta}\enspace,
\end{align}
where the final inequality follows from the bound on the potential outcomes, the bound on the weighting coefficients given in Lemma~\ref{lemma:bounded-var-coeff}, and the fact that the exposures take values in $[0,1]$.
Thus, the variance in the second term is at most
\[
\Var[\Big]{\poi(\zv) \poj(\zv) S_{i,j}(\dosei, \dosej)} 
\leq \paren[\Big]{24 \cdot \frac{M^2}{\Delta} }^2
\enspace.
\]
Plugging these two bounds into the bound on $\Var{\covestij}$, and using the Arithmetic-Geometric Inequality, we obtain
\[
\Var{\covestij} 
\leq \paren[\Bigg]{ 2 \maxpo^2 \paren[\Big]{\frac{\outdeg}{p(1-p)}}^2 + 24 \cdot \frac{M^2}{\Delta} }^2
\leq C \maxpo^4 \paren[\Bigg]{ \paren[\Big]{\frac{\outdeg}{p(1-p)}}^4 + \frac{1}{\Delta^2} }
\enspace. 
\qedhere
\]
\end{proof}

Finally, we are ready to prove Theorem~\ref{thm:consistent-var-estimation}, which establishes consistency rates for the variance estimator.
At a high level, we will combine Lemmas~\ref{lemma:covest-dep-degree-bound} and \ref{lemma:covest-dep-degree-bound}, which show that the individual covariance estimators are sufficiently uncorrelated and have small variance.

\begin{manualtheorem}{\ref{thm:consistent-var-estimation}}
\consistentvarestimation
\end{manualtheorem}

\begin{proof}
By unbiasesness of the variance estimator together a decomposition of its variance, we have that
\begin{align}
    \Exp[\Big]{\paren[\big]{n \cdot \Var{\dhtest} - n \cdot \erlvarest}^2}
    &= n^2 \cdot \Var[\Big]{\erlvarest} 
    = n^2 \cdot \Var[\Bigg]{ \frac{1}{n^2} \sum_{i=1}^n \sum_{j=1}^n \covestij } \\
    &= \frac{1}{n^2} \sum_{i=1}^n \sum_{j=1}^n \sum_{k=1}^n \sum_{\ell=1}^n \Cov{\covestij, \covestkl} \enspace.
    \intertext{
    We now discuss which terms are zero in the sum.
    Recall that when $\Cov{\dosei, \dosej} = 0$, then $\covestij = 0$. 
    Moreover, $\Cov{\dosei, \dosej} = 0$ for all $j \notin \depN{i}$. 
    Thus, $\Cov{\covestij, \covestkl} = 0$ for any $j \notin \depNi$.
    Additionally, the individual covariance estimators $\covestij$ and $\covestkl$ are uncorrelated if $(k,l) \notin \depNvar{i,j}$.
    Thus, we may simplify terms as}
    &= \frac{1}{n^2} \sum_{i = 1}^n \sum_{j \in \depN{i}} \sum_{(k, \ell) \in \depNvar{i,j}} \Cov{\covestij, \covestkl} \enspace.
    \intertext{Next, we use Cauchy-Schwarz inequality on the covariances together with the upper bound the variances $\Var{\covestij}$ provided by Lemma~\ref{lemma:mse-bound-on-covest} to obtain}
    &\leq \frac{1}{n^2} \sum_{i=1}^n  \sum_{j \in \depN{i}} \sum_{(k, \ell) \in \depNvar{i,j}} \sqrt{\Var{\covestij} \Var{\covestkl}} \\
    &\leq \frac{1}{n^2} \sum_{i=1}^n \sum_{j \in \depN{i}} \sum_{(k, \ell) \in \depNvar{i,j}} C \maxpo^4 \paren[\Bigg]{ \paren[\Big]{\frac{\outdeg}{p(1-p)}}^4 + \frac{1}{\Delta^2} } \enspace. \\
    \intertext{Finally, we use the maximum dependency degree bounds $\max_{i \in [n]} \abs{\depNi} \triangleq \depdeg \leq \maxcorr \divdeg \outdeg$ and $\max_{i,j \in [n]} \abs{\depNvar{i,j}} \triangleq \depdegvar \leq 4(\maxcorr \divdeg \outdeg)^2$ and rearrange terms to obtain}
    &\leq \frac{1}{n^2} n \depdeg \depdegvar \cdot  C \maxpo^4 \paren[\Bigg]{ \paren[\Big]{\frac{\outdeg}{p(1-p)}}^4 + \frac{1}{\Delta^2} } \\
    &\leq \frac{4 C}{n} \paren{\divdeg \outdeg}^3 \cdot \maxpo^4 \paren[\Bigg]{ \paren[\Big]{\frac{\outdeg}{p(1-p)}}^4 + \frac{1}{\Delta^2} } \enspace.
\end{align}

Under Assumptions~\ref{assumption:bounded-potential-outcomes} and \ref{assumption:design}, the magnitude of the potential outcomes $\maxpo$, the size of the clusters $\maxcorr$, and the term $1/p(1-p)$ are constants in the asymptotic sequence.
Thus, mean squared error of the normalized variance estimator is bounded as
\[
\Exp[\Big]{\paren[\big]{n \cdot \Var{\dhtest} - n \cdot \erlvarest}^2} 
= \bigO[\Big]{ \frac{1}{n} \cdot \paren[\Big]{\divdeg^3 \outdeg^7 + \frac{1}{\Delta^2}}}
\enspace.
\]
\end{proof}

\subsection{Asymptotic validity of confidence intervals (Theorem~\ref{thm:asymptotic-normality})}

In this section, we present the proof of Corollary~\ref{cor:asymptotic-validity} which establishes asymptotic validity of the Wald-type confidence intervals using the variance estimator.

\begin{lemma}\label{lemma:ratio-var-estimator}
Under Assumptions~\ref{assumption:bounded-potential-outcomes}, \ref{assumption:design}, \ref{assumption:mse-rate}, and \ref{assumption:non-degenerate-exposures} and further supposing that $\outdeg^3 \divdeg^7 = \littleO{n}$ and $\Delta = \omega \paren{n^{-1/2}}$, the ratio of the variance estimator and the true estimator converges to 1 in probability: $\frac{\Var{\dhtest}}{\erlvarest} \xrightarrow[]{p} 1$.
\end{lemma}

\begin{proof}
This may be shown by applying the continuous mapping theorem to the result of Theorem~\ref{thm:consistent-var-estimation}, as Assumption~\ref{assumption:mse-rate} bounds the normalized variance away from zero.
However, we take a more elementary approach using Chebyshev's inequality.

Let $\epsilon > 0$ be given.
Chebyshev's inequality states that $\Pr \paren{ \abs{X - \mu} \geq k \sigma }$ for any random variable $X$ with mean $\mu$ and standard deviation $\sigma$: .
For random variables with positive mean, rearranging terms yields $\Pr \paren{ \abs{\frac{X}{\mu} - 1} > \epsilon} \leq \frac{\sigma^2}{\epsilon^2 \mu^2}$.
Applying Chebyshev's inequality together with the bound on the mean squared error of the variance estimator (Theorem~\ref{thm:consistent-var-estimation}) and Assumption~\ref{assumption:mse-rate}, we have
\begin{align}
    \Pr \paren[\Bigg]{\abs[\Big]{ \frac{\Var{\dhtest}}{\erlvarest} - 1} > \epsilon}
    &\leq \frac{\Exp[\Big]{\paren[\big]{\Var{\dhtest} - \erlvarest}^2}}{\epsilon^2 \Var{\dhtest}^2} 
        &\text{(Chebyshev's Inequality)} \\
    &= \frac{n^2}{n^2} \cdot \frac{\Exp[\Big]{\paren[\big]{\Var{\dhtest} - \erlvarest}^2}}{\epsilon^2 \Var{\dhtest}^2} \\
    &= \frac{\Exp[\Big]{\paren[\big]{n \cdot \Var{\dhtest} - n \cdot \erlvarest}^2}}{\epsilon^2 \paren{n \cdot \Var{\dhtest}}^2} \\
    &\leq \frac{1}{\epsilon^2} \cdot \bigO[\Big]{ \frac{1}{n} \cdot \paren[\Big]{\divdeg^3 \outdeg^7 + \frac{1}{\Delta^2}}}     &\text{(Theorem~\ref{thm:consistent-var-estimation} and Assumption~\ref{assumption:mse-rate})}\\
    &= \frac{1}{\epsilon} \cdot \littleO{1} 
    \enspace,
\end{align}
where the final inequality follows from the assumptions that 
$\outdeg^3 \divdeg^7 = \littleO{n}$ and $\Delta = \omega \paren{n^{-1/2}}$. 
This establishes that the ratio of the variance estimator and the true estimator converges to 1 in probability.
\end{proof}

We are now ready to prove Corollary~\ref{cor:asymptotic-validity}, which establishes asymptotic validity of the Wald-based confidence intervals using the proposed variance estimator.
For completeness, we restate the corollary below.

\begin{manualcorr}{\ref{cor:asymptotic-validity}}
\asymptoticvalidity
\end{manualcorr}

\begin{proof}
Define the random variable $Z = \frac{\ate - \dhtest}{\sqrt{\Var{\dhtest}}}$.
By Theorem~\ref{thm:asymptotic-normality}, $Z$ converges in distribution to a standard normal, $Z \xrightarrow[]{d} \mathcal{N}(0,1)$.
Define $Z' = \frac{\ate - \dhtest}{\sqrt{\erlvarest}}$ and observe that 
\[
    Z' 
    = \frac{\ate - \dhtest}{\sqrt{\erlvarest}}
    = \frac{\ate - \dhtest}{\sqrt{\Var{\dhtest}}} \cdot \frac{\sqrt{\Var{\dhtest}}}{\sqrt{\erlvarest}}
    = Z \cdot \sqrt{ \frac{\Var{\dhtest}}{\erlvarest}}
\]
By Lemma~\ref{lemma:ratio-var-estimator}, the ratio of the variance and the variance estimator converges to 1 in probability.
Thus, by Slutsky's theorem, $Z' \xrightarrow[]{d} \mathcal{N}(0,1)$.

Now, we evaluate the probability of coverage in the limit.
By rearranging terms, we can re-write the coverage probability in terms of the tails of $Z'$ as follows:
\begin{align}
\lim_{n \rightarrow \infty} \Pr \paren[\Big]{ \ate \in \bracket[\Big]{ \dhtest \pm \Phi^{-1}(1 - \alpha/2) \sqrt{\erlvarest}}  }
&= \lim_{n \rightarrow \infty}
\Pr \paren[\Big]{ \Phi^{-1}(1 - \alpha/2) \leq \frac{\ate - \dhtest}{\sqrt{\erlvarest}} \leq \Phi^{-1}(1 - \alpha/2) } \\
&= \lim_{n \rightarrow \infty}
\Pr \paren[\Big]{ \Phi^{-1}(1 - \alpha/2) \leq Z' \leq \Phi^{-1}(1 - \alpha/2) } \\
&= \lim_{n \rightarrow \infty}
\Pr \paren[\Big]{ \Phi^{-1}(1 - \alpha/2) \leq Z' \leq \Phi^{-1}(\alpha/2) } \enspace,
\intertext{where the last equality follows from symmetry of the normal distribution. Let $F_n$ be the cumulative distribution function of $Z'$. By the convergence of $Z'$ in distribution to a standard normal, we have that}
&= \lim_{n \rightarrow \infty} F_n \paren[\Big]{\Phi^{-1}(1 - \alpha/2)} - F_n \paren[\Big]{\Phi^{-1}(\alpha/2) } \\
&= \Phi \paren[\Big]{\Phi^{-1}(1 - \alpha/2)} - \Phi \paren[\Big]{\Phi^{-1}(\alpha/2)} \\
&= (1 - \alpha/2) - (\alpha / 2) \\
&= 1 - \alpha 
\enspace. \qedhere
\end{align}
\end{proof}
\section{\exposuredesign and Correlation Clustering}\label{sec:design-proofs}

In this section, we prove the relationship between \exposuredesign, its reformulation \corrclust, the previously proposed correlation clustering design of \cite{pouget2019variance}, and other correlation clustering variants.
A summary of the results are:
\begin{itemize}
    \item In Section~\ref{sec:reformulation}, we show that the \exposuredesign may be reformulated as the clustering problem, \corrclust.
    \item In Section~\ref{sec:prev_instance}, we compare \exposuredesign to the correlation clustering-based design presented in \cite{pouget2019variance}. In particular, we prove that their design is equivalent to \exposuredesign when the trade-off parameter is set as $\phi = 1/(n-1)$ and no constraint is placed on cluster sizes, i.e. $\maxcorr = m$.
    \item In Section~\ref{sec:comparing_variants}, we compare \corrclust to other correlation clustering variants. In particular, we prove that (unconstrained) \corrclust may be viewed as an instance of the weighted maximization correlation clustering considered by \cite{CGV05clustering, Swamy04correlation} but with a possibly large additive constant which prevents an approximation-preserving reduction.
\end{itemize}

To begin, we demonstrate how to re-write the \corrclust objective using matrix notation.
Let $\omega_{i,j} \in \Reals$ be the weights for pairs $i,j \in [m]$ and let $\Omega$ be the $m$-by-$m$ matrix whose $(i,j)$th entry is $\omega_{i,j}$.
For a partition $\clustering$ of the indices $[m]$, let $Z_\clustering$ be the $m$-by-$m$ matrix where the $(i,j)$th entry is $1$ if $i$ and $j$ are in the same cluster of $\clustering$ and $0$ otherwise.
Then, we may express the \corrclust objective as
\[
\sum_{\cluste{r} \in \clustering} \sum_{i,j \in \cluste{r}} \omega_{i,j}
=
\sum_{i=1}^n \sum_{j=1}^n \omega_{i,j} \bracket{Z_\clustering}_{(i,j)}
= \tr{ \Omega \ Z_\clustering } \enspace.
\]
Throughout the remainder of the section, it will be useful to write the \corrclust objective using this matrix notation.

\subsection{Reformulating \exposuredesign as \corrclust}\label{sec:reformulation}
We are now ready to prove Proposition~\ref{prop:corr_cluster_instance}, which we restate here for completeness.

\begin{manualprop}{\ref{prop:corr_cluster_instance}}
For each pair of diversion units $i,j \in \divunits$, define the value $\omega_{i,j} \in \Reals$ as
\begin{equation} \tag{\ref{eq:corr-clustering-weights}}
\omega_{i,j} 
= \paren{1 + \phi} \sum_{k=1}^m w_{k,i} w_{k,j}
- \phi \paren[\Big]{\sum_{k=1}^m w_{k,i}} \paren[\Big]{\sum_{k=1}^m w_{k,j}} \enspace,
\end{equation}
where $w_{k,i}$ is the weight of the edge between the $k$th outcome unit and the $i$th diversion unit.
\ref{eq:cluster_objective} is equivalent to the following clustering problem:
\begin{equation}\tag{\corrclust}
\max_{\text{clusterings } \clustering}
\quad \sum_{\cluste{r} \in \clustering} \sum_{i, j \in \cluste{r}} \omega_{i,j} 
\enspace.
\end{equation}
\end{manualprop}
\begin{proof}
Recall that the objective of \exposuredesign is defined as
\[
\sum_{i=1}^n \Var{\dosei} - \phi \sum_{i \neq j} \Cov{ \dosei, \dosej} \enspace,
\]
where the expectation in the variance and covariance terms is taken with respect to the random assignment vector $\zv \in \setb{0, 1}^m$, which is drawn from the cluster design given by $\clustering$.
Recall that the exposures are given by $\dosev = \weightM \zv$.
Using matrix notation, we can more compactly represent this objective as
\begin{align*}
\sum_{i=1}^n \Var{\dosei} - \phi \sum_{i \neq j} \Cov{ \dosei, \dosej}
&= \tr[\Big]{ \paren[\big]{\idM - \phi \paren{\onevec \onevec^\tran - \idM}} \Cov{\dosev} } 
    &\text{(rewriting in terms of $\trm$)}\\
&= \tr[\Big]{ \paren[\big]{\paren{1+\phi}\idM - \phi \onevec \onevec^\tran} \Cov{\dosev} }
    &\text{(rearranging terms)} \\
&= \tr[\Big]{ \paren[\big]{\paren{1+\phi}\idM - \phi \onevec \onevec^\tran} \Cov{\weightM \zv} }
    &\text{(definition of exposure)} \\
&= \tr[\Big]{ \paren[\big]{\paren{1+\phi}\idM - \phi \onevec \onevec^\tran} \weightM \Cov{\zv} \weightM^\tran }
    &\text{(property of covariance)} \\
&= \tr[\Big]{ \weightM^\tran \paren[\big]{\paren{1+\phi}\idM - \phi \onevec \onevec^\tran} \weightM \Cov{\zv}}
    &\text{(cyclic property of trace)} \\
\end{align*}
Because $\zv$ is drawn from an independent cluster design, the $(i,j)$th entry of the covariance matrix $\Cov{\zv}$ is $1/2$ if diversion units $i$ and $j$ are in the same cluster and $0$ otherwise.
Thus, by the observation above, this clustering objective is a correlation clustering where the weights are given by the matrix
\[
\Omega = \weightM^\tran \paren[\big]{\paren{1+\phi}\idM - \phi \onevec \onevec^\tran} \weightM \enspace.
\]
By inspection, we have that the $(i,j)$th entry of this matrix $\Omega$ is 
\[
\omega_{i,j} 
= \paren{1 + \phi} \sum_{k=1}^n w_{k,i} w_{k,j}
- \phi \paren[\Big]{\sum_{k=1}^n w_{k,i}} \paren[\Big]{\sum_{k=1}^n w_{k,j}} \enspace,
\]
as desired.
\end{proof}

\subsection{An instance of \exposuredesign when $\phi = 1/(n-1)$}\label{sec:prev_instance}

Now we demonstrate that the correlation clustering objective proposed in \cite{pouget2019variance} is a special case of \exposuredesign when $\phi = 1/(n-1)$ and no constraint is placed on cluster sizes, i.e. $\maxcorr = m$.
Before giving the formal statement, we re-introduce the clustering objective in that paper; that is,
\begin{equation}\label{eq:exposure_spread_obj}\tag{\textsc{Exposure-Spread}\xspace}
\max_{\text{clusterings } \clustering} 
\Exp[\Bigg]{\sum_{i=1}^n \paren[\Bigg]{\dosei - \paren[\Big]{ \frac{1}{n} \sum_{j=1}^n \dosej} }^2}
\enspace,
\end{equation}
where the expectation is with respect to the treatment vector $\zv \in \setb{\pm 1}^m$ drawn according to the independent cluster design given by $\clustering$.
The quantity in the expectation is a measure of the spread of the exposures.
We remark that in \cite{pouget2019variance}, the exposures are called ``doses'' and the quantity in the expectation is referred to as the ``empirical dose variance''.

\begin{prop}\label{prop:more-general-clustering-obj}
Up to additive and multiplicative constants, \ref{eq:exposure_spread_obj} is equivalent to \exposuredesign when the trade-off parameter is set to $\phi = 1/(n-1)$.
\end{prop}

\begin{proof}
Let us denote the exposure spread by
\[
Q = \sum_{i=1}^n \paren[\Bigg]{\dosei - \paren[\Big]{ \frac{1}{n} \sum_{j=1}^n \dosej} }^2 \enspace,
\]
Note that the exposure spread is equal to the $\ell_2$ norm of the \emph{de-meaned} exposure vector $\bar{\dosev} = \paren{\bar{\dosee{1}}, \bar{\dosee{2}}, \dots \bar{\dosee{n}}}$, where
\[
\bar{\dosei} = \dosei - \paren[\Big]{ \frac{1}{n} \sum_{j=1}^n \dosej} \enspace .
\]
The entire de-meaned exposure vector may be written as $\bar{\dosev} = \paren{\idM - \frac{1}{n} \onevec \onevec^T} \dosev$.
Using the fact that this matrix is a projection and that the exposure vector is $\dosev = \weightM \zv$, we can write the exposure spread as 
\begin{align*}
Q
= \| \bar{\dosev} \|^2 
= \| \paren{\idM - \frac{1}{n} \onevec \onevec^T} \dosev \|^2 
= \dosev^\tran \paren{\idM - \frac{1}{n} \onevec \onevec^T}^2 \dosev 
= \dosev^\tran \paren{\idM - \frac{1}{n} \onevec \onevec^T} \dosev 
= \zv^\tran \weightM^\tran \paren{\idM - \frac{1}{n} \onevec \onevec^T} \weightM \zv \enspace.
\end{align*}
Finally, the expectation of the exposure spread may be written as
\begin{align*}
\Exp{Q}
&= \Exp[\Bigg]{ \zv^\tran \weightM^\tran \paren{\idM - \frac{1}{n} \onevec \onevec^T} \weightM \zv } 
    &\text{(from above)}\\
&= \Exp[\Bigg]{ \tr[\Big]{ \zv^\tran \weightM^\tran \paren{\idM - \frac{1}{n} \onevec \onevec^T} \weightM \zv }} 
    &\text{(trace of a scalar)}\\
&= \Exp[\Bigg]{ \tr[\Big]{ \weightM^\tran \paren{\idM - \frac{1}{n} \onevec \onevec^T} \weightM \zv \zv^\tran }} 
    &\text{(cyclic property of trace)}\\
&= \tr[\Big]{ \weightM^\tran \paren{\idM - \frac{1}{n} \onevec \onevec^T} \weightM \Exp{ \zv \zv^\tran }}
    &\text{(linearity of trace)} \\
&= \tr[\Big]{ \weightM^\tran \paren{\idM - \frac{1}{n} \onevec \onevec^T} \weightM \Cov{\zv}} + c \enspace,
\end{align*}
where the value $c$ in the last line is $c = \tr[\Big]{ \weightM^\tran \paren{\idM - \frac{1}{n} \onevec \onevec^T} \weightM \Exp{ \zv} \Exp{ \zv}^\tran }$, which follows from $\Cov{\zv} = \Exp{\zv \zv^\tran} - \Exp{\zv} \Exp{\zv}^\tran$ and linearity of trace.
Moreover, when the probability of treatment assignment $p$ is fixed, this value $c$ is a constant with respect to the clustering being chosen. 

Observe that by setting $\phi = 1/(n-1)$ and multiplying by a factor $(n-1)/n$, the \exposuredesign objective becomes
\[
\frac{n-1}{n} \tr[\Big]{ \weightM^\tran \paren[\big]{\paren{1+\frac{1}{n-1}}\idM - \frac{1}{n-1} \onevec \onevec^\tran} \weightM \Cov{\zv}}
=
\tr[\Big]{ \weightM^\tran \paren{\idM - \frac{1}{n} \onevec \onevec^T} \weightM \Cov{\zv}} \enspace.
\]
Thus, the \ref{eq:exposure_spread_obj} objective is equivalent (up to additive and multiplicative constants) to the \exposuredesign objective when $\phi = 1/(n-1)$.
\end{proof}

\subsection{Comparison to other correlation clustering variants}\label{sec:comparing_variants}

Recall that we defined the objective of the correlation clustering variant \corrclust as
$$
\sum_{\cluste{r} \in \clustering} \sum_{i, j \in \cluste{r}} \omega_{i,j} 
\enspace,
$$
where $\omega_{i,j}$ is defined for each pair of diversion units $i,j \in \divunits$ as
\begin{equation*} 
\omega_{i,j} 
= \paren{1 + \phi} \sum_{k=1}^n w_{k,i} w_{k,j}
- \phi \paren[\Big]{\sum_{k=1}^n w_{k,i}} \paren[\Big]{\sum_{k=1}^n w_{k,j}} \enspace,
\end{equation*}
and $w_{k,i}$ is the weight of the edge between the $k$th outcome unit and the $i$th diversion unit.
Observe that the term $\omega_{i,j}$ can take positive or negative values.

The maximization weighted correlation clustering variant considered by \cite{CGV05clustering, Swamy04correlation} is defined as follows. 
Let $G = (V,E)$ be a graph where each edge $e = (i,j) \in E$ has two \emph{non-negative} weights: $w_{in}(i,j)$ and $w_{out}(i,j)$.
Given a clustering $\clustering$, an edge $e = (i,j)$ is said to be \emph{in-cluster} if $i$ and $j$ are in the same cluster and \emph{out-cluster} otherwise.
The objective function for a given clustering is given by
\begin{equation}\tag{\corrclustCS}
    \sum_{\substack{\text{in-cluster} \\ \text{edges } e}} w_{in}(e) 
    + \sum_{\substack{\text{out-cluster} \\ \text{edges } e}} w_{out}(e)
\end{equation}
We now show that the \corrclust objective may be written as an instance of the \corrclustCS objective, but with the addition of a large additive constant.
Again, we stress that this reduction is primarily for aesthetic comparison purposes because the appearance of the large additive constant prevents any meaningful approximation-preserving reduction.

\begin{prop}
Our formulation \corrclust may be viewed as an instance of \corrclustCS with a large additive constant.
More precisely, let $w_{in}(i,j) = \max\{0, \omega_{i,j} \}$ and $w_{out}(i,j) = \min \{0, \omega_{i,j} \}$.
For a clustering $\clustering$, we have that the objectives are related by
\[
\sum_{\cluste{r} \in \clustering} \sum_{i, j \in \cluste{r}} \omega_{i,j} 
- \sum_{i=1}^n \sum_{j=1}^n \min \{0, \omega_{i,j} \}
= \sum_{\substack{\text{in-cluster} \\ \text{edges } e}} w_{in}(e) 
    + \sum_{\substack{\text{out-cluster} \\ \text{edges } e}} w_{out}(e)
\]
\end{prop}
\begin{proof}
For each pair of diversion units $i,j$, define $\omega_{i,j}^+ = \max\{0, \omega_{i,j}\}$ and $\omega_{i,j}^- = -\min\{0, \omega_{i,j}\}$.
Observe that $\omega_{i,j} = \omega_{i,j}^+ + \omega_{i,j}^-$ and so we can re-distribute the following sum as
\[
\sum_{\cluste{r} \in \clustering} \sum_{i, j \in \cluste{r}} \omega_{i,j} 
= \sum_{\cluste{r} \in \clustering} \sum_{i, j \in \cluste{r}} \paren[\Big]{\omega_{i,j}^+ + \omega_{i,j}^-}
= \sum_{\cluste{r} \in \clustering} \sum_{i, j \in \cluste{r}} \omega_{i,j}^+
+ \sum_{\cluste{r} \in \clustering} \sum_{i, j \in \cluste{r}} \omega_{i,j}^-
\enspace.
\]
Subtracting the (instance-dependent) constant $\sum_{i=1}^n \sum_{j=1}^n \min \{0, \omega_{i,j} \}$ from both sides and rearranging yields
\begin{align*}
\sum_{\cluste{r} \in \clustering} \sum_{i, j \in \cluste{r}} \omega_{i,j}
    - \sum_{i=1}^n \sum_{j=1}^n \min \{0, \omega_{i,j} \}
&= \sum_{\cluste{r} \in \clustering} \sum_{i, j \in \cluste{r}} \omega_{i,j}^+
    + \sum_{\cluste{r} \in \clustering} \sum_{i, j \in \cluste{r}} \omega_{i,j}^-
    - \sum_{i=1}^n \sum_{j=1}^n \min \{0, \omega_{i,j} \} \\
&= \sum_{\cluste{r} \in \clustering} \sum_{i, j \in \cluste{r}} \omega_{i,j}^+
    - \sum_{\cluste{r} \neq \cluste{r}' \in \clustering} \sum_{\substack{i \in \cluste{r} \\ j \in \cluste{r}'}} \omega_{i,j}^- \\
&= \sum_{\cluste{r} \in \clustering} \sum_{i, j \in \cluste{r}} \omega_{i,j}^+
    + \sum_{\cluste{r} \neq \cluste{r}' \in \clustering} \sum_{\substack{i \in \cluste{r} \\ j \in \cluste{r}'}} \paren{-\omega_{i,j}^-} \\
&= \sum_{\substack{\text{in-cluster} \\ \text{edges } e}} w_{in}(e) 
    + \sum_{\substack{\text{out-cluster} \\ \text{edges } e}} w_{out}(e) 
    \enspace.
\end{align*}
Finally, observe that for each pair $(i,j)$, the values $w_{in}(i,j)$  and $w_{out}(i,j)$ are non-negative so that the final equation is a valid objective function for the \corrclustCS formulation.
\end{proof}

\end{document}